\documentclass[11pt]{scrartcl}
\KOMAoptions{toc=sectionentrywithdots}
\usepackage[a4paper, total={16cm, 24cm}]{geometry}
\usepackage{amsmath,amsthm,amssymb,mathtools}
\usepackage{microtype}
\usepackage{thmtools}
\makeatletter
\ifdefined\newcounteralias
  \renewcommand\thmt@autorefsetup{\expandafter\edef\csname\thmt@envname autorefname\endcsname
      {\thmt@thmname}}
\fi
\makeatother
\usepackage{caption}
\usepackage{titletoc}
\usepackage{wrapstuff}
\usepackage{subcaption}
\usepackage{enumitem}
\usepackage{bm}
\usepackage{tikz}
\usetikzlibrary{arrows.meta, calc}
\usetikzlibrary{patterns.meta}
\usepackage{pgfplots}
\pgfplotsset{compat=1.18}

\usepackage{booktabs} \usepackage{xcolor}   

\usepackage{todonotes}
\usepackage[backend=biber,style=authoryear,natbib=true,maxbibnames=99,maxcitenames=2,backref=true]{biblatex}
\usepackage{hyperref}
\hypersetup{
	pdfencoding=auto,
	psdextra,
	colorlinks=true,
	citecolor=green!50!black,
	linkcolor=red!60!black,
	urlcolor=blue!90!black
}
\usepackage{authblk}

\usepackage[ruled,vlined]{algorithm2e}
\SetKwComment{Comment}{$\triangleright$\ }{}
\SetKwFunction{ApproxValue}{ApproxValue}
\SetKwProg{Fn}{Function}{:}{}

\newtheorem{theorem}{Theorem}[section]
\newtheorem{lemma}[theorem]{Lemma}

\newtheorem{proposition}[theorem]{Proposition}
\theoremstyle{definition}
\newtheorem{definition}[theorem]{Definition}
\newtheorem{example}[theorem]{Example}
\theoremstyle{remark}

\usepackage[nameinlink]{cleveref}
\AddToHook{cmd/appendix/before}{\crefalias{section}{appendix}\crefalias{subsection}{appendix}}
\crefname{appendix}{Appendix}{Appendices}
\Crefname{appendix}{Appendix}{Appendices}

\crefname{algocf}{Algorithm}{Algorithms}
\Crefname{algocf}{Algorithm}{Algorithms}

\crefname{theorem}{Theorem}{Theorems}
\Crefname{theorem}{Theorem}{Theorems}

\crefname{lemma}{Lemma}{Lemmata}
\Crefname{lemma}{Lemma}{Lemmata}

\crefname{corollary}{Corollary}{Corollaries}
\Crefname{corollary}{Corollary}{Corollaries}

\crefname{proposition}{Proposition}{Propositions}
\Crefname{proposition}{Proposition}{Propositions}

\crefname{fact}{Fact}{Facts}
\Crefname{fact}{Fact}{Facts}

\crefname{definition}{Definition}{Definitions}
\Crefname{definition}{Definition}{Definitions}

\crefname{example}{Example}{Examples}
\Crefname{example}{Example}{Examples}

\crefname{remark}{Remark}{Remarks}
\Crefname{remark}{Remark}{Remarks}

\DeclareMathOperator*{\argmax}{arg\,max}
\newcommand{\R}{\mathbb{R}}

\newcommand{\N}{\mathbb{N}}
\DeclareMathOperator{\supp}{supp}
\newcommand{\pav}[1][]{\ifthenelse{\equal{#1}{}}{\ensuremath{\mathit{PAV}}\xspace}{\ensuremath{\mathit{PAV}(#1)}\xspace}}
\newcommand{\mes}[1][]{\ifthenelse{\equal{#1}{}}{\ensuremath{\mathit{MES}}\xspace}{\ensuremath{\mathit{MES}(#1)}\xspace}}

\title{Existence of the Core in \\ Approval-Based Committee Elections}

\author[1]{Patrick Becker}
\author[2]{Matthias Greger}
\author[3]{Dominik Peters}
\affil[1]{Technical University of Munich}
\affil[2]{University of Oxford}
\affil[3]{CNRS, LAMSADE, Universit\'e Paris Dauphine - PSL}
\date{\vspace{-1cm}}

\begin{document}

\maketitle

\begin{abstract}
	\noindent We resolve the central open question in the theory of approval-based committee elections: we prove that the core is always non-empty. This means that there always exists a committee for which no coalition of voters can deviate to an alternative committee of proportional size that every member strictly prefers. This property was introduced by \citet{ABC+16a} to capture proportional representation and coalitional stability, but no known algorithm achieved it except up to a constant approximation factor or in special cases.

	Our approach begins by strengthening the core to \emph{core+}, which requires the committee to be robust even to fractional deviations. Using LP duality, this property can equivalently be defined in terms of a payment-based certificate under which no losing candidate gathers too much support. We then define a voting rule that chooses a committee admitting a certificate that maximizes a novel \emph{harmonic entropy} objective function, which rewards balanced payment distributions. We show that the payment system associated with an optimal committee certifies that the committee satisfies core+, establishing existence. The same is true for local optima, which allows us to design a polynomial-time algorithm computing a core+ committee.
\end{abstract}

\vspace{0.2cm}
\begingroup
    \renewcommand*{\raggedsection}{\centering}
    \setcounter{tocdepth}{1}
	\linespread{0.7}\selectfont
	\tableofcontents
\endgroup
\newpage

\section{Introduction}
Approval-based committee elections concern the problem of selecting a fixed-size committee $W$ of $k \in \N$ candidates from a set $C$ of $m$ candidates, given the preferences of a set $N$ of $n$ voters. Each voter $i\in N$ reports a set of approved candidates $A_i\subseteq C$, and derives utility $u_i(W)=\lvert A_i\cap W \rvert$ from a committee $W \subseteq C$. 

A central objective in this setting is \emph{proportional representation}: sufficiently large groups of voters with similar preferences should receive representation proportional to their size.
One of the strongest formalizations of this principle is \emph{core stability}, which was introduced to approval-based committee elections by \citet{ABC+16a} and has its roots in the classical core of cooperative game theory; see, e.g., \citet{shapley1955markets, Gill59a, Scar67a}. The core formalizes stability by allowing each coalition $S\subseteq N$ to use its proportional share of the available seats to propose an alternative set of candidates $T$. The \emph{quota} $q = \frac{n}{k}$ thereby specifies the number of voters necessary to ``deserve'' one candidate. With this, a committee $W \subseteq C$ is in the core if no coalition $S$ can propose another set $T \subseteq C$ that it can afford, i.e., $\lvert S \rvert \ge q \, \lvert T \rvert$, and that every coalition member $i \in S$ strictly prefers to $W$, i.e., $u_i(T) \ge u_i(W) + 1$. Such a coalition $S$ has a justified complaint of underrepresentation.

A variety of voting rules have been proposed in this setting. While some of them get close to satisfying the core (in the sense of approximating the core or satisfying relaxed versions of it), counterexamples are known for all of them. Thus, it has long been a central open problem of the field whether there always exists a committee in the core  \citep{ABC+16a,LaSk23a}. Much of the literature on proportional representation in approval-based committee elections can be viewed as approaching this question from different directions, and we survey prior attempts below. In this paper, we resolve the existence question positively by defining a new voting rule, called the \emph{harmonic entropy rule}, that always selects committees in the core.

Our approach is based on a strengthening of the core which we call \emph{core+}, since it is analogous to the other recent axioms of $\mathrm{FJR+}$ \citep{teh2026strengthening} and $\mathrm{EJR+}$ \citep{BrPe23a}, which strengthen the classic proportionality axioms of \emph{full justified representation} ($\mathrm{FJR}$) \citep{PPS21a} and \emph{extended justified representation} ($\mathrm{EJR}$) \citep{ABC+16a} in a similar way. Note that the core implies both $\mathrm{FJR}$ and $\mathrm{EJR}$. 
For a committee to satisfy the core+ property, it is not allowed to even admit \emph{fractional} deviations, where voters can fractionally participate in a blocking coalition and also propose a fractional set of candidates; see \Cref{def:core+-deviation} for the precise definition.

Given its fractional nature, the existence of a core+ deviation can be checked by a linear feasibility problem. By LP duality, the absence of such a deviation admits an equivalent characterization in terms of a \emph{payment system}. In this system, each voter $i \in N$ makes payments $p_{ic} \ge 0$ to approved committee members $c \in A_i \cap W$, subject to the following four conditions:
\begin{itemize}
	\item \ref{cond:p1}: Every voter $i \in N$ has a unit budget, where the unspent budget $r_i = 1 - \sum_{c \in A_i \cap W} p_{ic} \ge 0$ is called the \emph{reserve} of $i \in N$.
	\item \ref{cond:p2}: The reserve is at least as large as any individual payment, that is, $r_i \ge p_{ic}$.
	\item \ref{cond:p3}: The accumulated reserve of the supporters of a losing candidate $c \in C \setminus W$ is strictly less than the quota $q$.
	\item \ref{cond:p4}: The total payments to each winning candidate $c \in W$ are at most the quota $q$.
\end{itemize}
Thus, we have reduced the task to showing that there always exists a committee that can be equipped with a payment system satisfying \ref{cond:p1}--\ref{cond:p4}.

Payment-based approaches to designing proportional voting rules have been used before, but have so far fallen short of establishing core existence. 
The \emph{Method of Equal Shares} (\mes), for example, satisfies \ref{cond:p1}, \ref{cond:p3}, and \ref{cond:p4} but voters may spend their entire budget, leaving no reserve \citep{PeSk20a}.
We can observe an inherent trade-off between \ref{cond:p2} and \ref{cond:p3}: voters are forced to retain a reserve at least as large as their highest payment due to \cref{cond:p2}, while \ref{cond:p3} requires candidates outside the committee not to be able to accumulate enough of voters' reserves. That said, it turns out that \emph{Proportional Approval Voting} (\pav), a well-studied voting rule originally proposed by \citet{Thie95a}, can be seen as resolving this tension.
The rule is based on selecting committees $W$ maximizing the sum of voters' harmonic utilities, that is,
\[
	\sum_{i \in N} 1 + \frac12 + \frac13 + \dots + \frac{1}{u_i(W)}.
\]
Notably, the definition of \pav does not contain any visible payments, and to our knowledge \pav has never been analyzed from a payments perspective. However, we show that an optimal \pav committee can always be equipped with the \emph{uniform payment system} where each voter pays $\frac{1}{u_i(W) + 1}$ to each approved winner and retains $\frac{1}{u_i(W) + 1}$ in reserve. The resulting system satisfies \ref{cond:p1}, \ref{cond:p2}, and \ref{cond:p3}, which follows from a classic swap argument that is often used to analyze \pav. 

While this rule satisfies the core for committee sizes up to 8 when breaking ties appropriately \citep{peters2025fewseats}, its payment system can unfortunately fail \ref{cond:p4}. Due to the rigidity of the uniform payment system, voters might spend too much on winners. This motivates allowing non-uniform payments. Nevertheless, our core+ existence argument is inspired by \pav and its established proof techniques. We seek payment systems that are \emph{as uniform as possible} while still satisfying \ref{cond:p4} and construct our voting rule in a way that allows us to recover the swap argument.

With this in mind, we propose the \emph{Harmonic Entropy Rule} which maximizes a carefully designed objective function over both size-$k$ committees and their payment systems. 
Like Shannon entropy, our objective favors uniform budget distributions. We call it \emph{harmonic} entropy because its value on the uniform payment distribution is a harmonic number.
For each committee $W$, we maximize total harmonic entropy over payment systems subject to \ref{cond:p1}, \ref{cond:p2}, and \ref{cond:p4}; we then show that if the supporters of any losing candidate collectively hold at least $q$ in reserve, a single-candidate swap strictly increases the objective. This gives \ref{cond:p3} and shows that the harmonic entropy rule satisfies core+. Hence, the core is non-empty.

Because our proof is based on a swap argument, local optima of our harmonic entropy objective also satisfy core+. In fact, it suffices for the argument that no swap increases the objective by at least an inverse-polynomial amount, which implies that a local search converges in polynomial time to a core+ committee. Because our objective function can be evaluated in polynomial time using a linear program, this yields a polynomial-time algorithm for computing a core+ committee.

Our proofs work for the quota $q = \frac{n}{k}$, which is known as the \emph{Hare quota}, but the same argument works for all $\frac{n}{k+1} < q \le \frac{n}{k}$. Via a limiting argument, it follows that even a committee satisfying the stronger property of \emph{Droop core+}, defined at the Droop quota $q = \frac{n}{k+1}$, exists (for an exact definition, see \Cref{section:droop}), though we do not know if such a committee can be found in polynomial time. That said, local search still yields a polynomial-time algorithm for the ordinary Droop core.

Our theorem on core existence for the Hare quota has been formally verified in Lean as part of the \href{https://github.com/DominikPeters/ABCVotingLean}{ABCVotingLean} project.

\paragraph{Related Work.}
As noted above, a lot of work on representation in approval-based committee elections can be viewed as having aimed to make progress on the core existence problem. This work can be categorized into four types of approaches.

First, one can weaken the core by restricting the coalitions that can raise objections to a committee. A natural restriction is \emph{cohesiveness}, which requires voters in a coalition to share sufficiently many approved candidates. This idea underlies several well-studied proportionality axioms, including \emph{justified representation} ($\mathrm{JR}$), $\mathrm{EJR}$, and $\mathrm{FJR}$ (see \citet{ABC+16a, PPS21a}).
Recent work has substantially strengthened the understanding of these notions: \citet{frank2026polynomial} and \citet{ai2026justifiedrepresentationharedroop} independently derived a polynomial-time algorithm to satisfy $\mathrm{FJR}$, while \citet{teh2026strengthening} introduces $\mathrm{FJR+}$, a strengthening of $\mathrm{FJR}$ that admits both polynomial-time verification and computation. 

Second, several papers have relaxed the core by allowing approximations of two different types \citep{fain2018fair,munagala2022approximate, cheng2020group}. \emph{Utility approximations} allow deviations only when the deviating coalition can meet stronger utility requirements, while \emph{group-size approximations} keep the utility requirements unchanged and only allow larger coalitions to deviate. For the former, \citet{PeSk20a} show that \pav gives a factor-2 utility approximation and \mes one in $O(\log k)$.
For the latter, \citet{jiang2020approximately} provide a $32$-approximation for arbitrary monotone preferences, which improves to $16$ in the approval case.
\citet{gao2025computation} substantially improve this bound to $3.651$. Their approach builds on rounding a Lindahl equilibrium with a strongly Rayleigh distribution.
Very recently, \citet{lin2026sequentialphragmenguarantees2approximate} show that the sequential Phragmén rule achieves a $2$-approximation.

Third, there has also been substantial progress on exact core existence in restricted settings. \citet{peters2025fewseats} proves that \pav always selects at least one committee in the core when the committee size is at most $8$ and proves existence when the number of candidates is at most $15$, while \citet{becker2026core} prove existence for elections with up to seven voter types and \citet{goel2026nash} with up to eight voters.
\citet{pierczynski2022core} obtain core existence for one-dimensional restricted domains such as candidate interval, while \citet{BGP+24a} show that \pav satisfies the core when every candidate has at least $k$ copies.

Finally, instead of discrete (indivisible) committee elections, one can study a continuous (divisible) model \citep[e.g.][]{ABM20a,SuVo24a}. Indeed, for the problem of selecting fractional committees, the existence of so-called Lindahl equilibrium due to \citeauthor{Fole70a}'s \parencite*{Fole70a} fixed-point argument implies the existence of fractional committees in the core \citep{FGM16a,munagala2022approximate,becker2026core}. \citet{KrPe25b} recently discovered a convex program whose optimum solutions are fractional committees in the core. Interestingly, speaking loosely, this program computes a fractional committee that can be equipped with a payment system that maximizes Shannon entropy, which is reminiscent of our harmonic entropy approach.

\paragraph{Outline.} We introduce core+ in \Cref{sec:core+}, connecting fractional core deviations to payment systems and the conditions \ref{cond:p1}--\ref{cond:p4}. In \Cref{section:pav}, we revisit the proportionality guarantees of \pav through the lens of these payment conditions. This analysis reveals the limitations of uniform payment systems and motivates the study of harmonic entropy in \Cref{section:harmonic-entropy}. In \Cref{section:harmonic-entropy-rule}, we establish core+ existence by analyzing how harmonic entropy changes when candidates are added to and removed from a committee. The resulting swap argument forms the basis of our polynomial-time algorithm in \Cref{section:computation}. Finally, in \Cref{section:droop}, we extend our analysis to the Droop core+.

\section{Preliminaries}

For $t \in \N$, define $[t] \coloneqq \{1, \dots, t\}$. For $a \in \R$, let $(a)_+ = \max\{a, 0\}$ be its positive part.
Let $\Delta^d \coloneqq \{ x\in\R_{\ge 0}^d : \sum_{j=1}^d x_j=1 \}$ denote the probability simplex on $d$ coordinates. We say that a distribution $p \in \Delta^d$ is \emph{strictly positive} if $p_j > 0$ for all $j \in [d]$.

Let $N=[n]$ be a set of \emph{voters}, let $C = \{c_1, \dots, c_m\}$ be a finite set of \emph{candidates}, and denote the \emph{committee size} by $k \in \N$ with $k \in [m-1]$.
For $c \in C$, we write $N_c = \{ i \in N : c \in A_i \}$ for the set of voters approving $c$.
A \emph{committee} is a candidate set $W \subseteq C$ and we say that $c$ is a winner if $c \in W$ and a loser if $c \not\in W$. 
We denote the set of all committees of size at most $k$ by $\mathcal{W}$ and the set of all committees of size exactly $k$ by $\mathcal{W}_k$.
Each voter $i$ has an individual set $A_i \subseteq C$ of approved candidates. Her utility $u_i(W)$ for a candidate set $W \subseteq C$ equals the number of her approved candidates in $W$, i.e., $u_i(W)=\lvert A_i\cap W \rvert$. 
An approval-based committee election \emph{instance} is denoted by $I = (A, k)$ where $A = (A_i)_{i \in N}$ denotes the approval profile. A \emph{quota} $q>0$ specifies how many voters are needed to ``deserve'' one candidate. The \emph{Hare quota} is $\frac{n}{k}$, and the \emph{Droop quota} is $\frac{n}{k+1}$. Unless stated otherwise, we assume $\frac{n}{k+1} < q \le \frac{n}{k}$.

\section{Core+ and Payment Systems} \label{sec:core+}

We first define the core and introduce a strengthening, which we will call \emph{core+}, obtained by allowing fractional deviations both with respect to voters and candidates. This follows the approach underlying $\mathrm{FJR+}$ \citep{teh2026strengthening}, which strengthens $\mathrm{FJR}$ through a similar relaxation.\footnote{We show in \Cref{prop:ejr+-connection-ejr} that $\mathrm{EJR+}$ and $\mathrm{EJR}$ admit a similar connection.}

A non-empty coalition $S\subseteq N$ \emph{blocks} a committee $W$ if there exists a proposal $T\subseteq C$ that the coalition can afford, $\lvert S \rvert \ge q \cdot \lvert T\rvert$, and that every member strictly prefers to $W$, i.e.,
\begin{equation} \label{eq:coredeviation}
	u_i(T) \ge u_i(W)+1 \qquad \text{for every } i\in S.
\end{equation}
A committee $W \in \mathcal W$ is \emph{in the core} if it is not blocked by any non-empty coalition $S\subseteq N$. Deciding whether a given committee is in the core is $\mathrm{coNP}$-complete \citep{BGP+24a}, hinting at the intricate structure of the core. Strengthening the core to core+ by enlarging the set of blocking coalitions simplifies that structure. 

\begin{definition}[Core+] \label{def:core+-deviation}
 	For a committee $W \subseteq C$, a \emph{fractional blocking coalition} consists of a voter participation vector $y \in \mathbb{R}^n_{\ge 0} \setminus \{\mathbf 0\}$ and candidate quantities $z \in \mathbb{R}^m_{\ge 0}$ such that the coalition can afford the candidates: $\sum_{i \in N} y_i \ge q\sum_{c \in C} z_c$, and every member strictly prefers $z$ to $W$:
	\begin{equation}\label{eq:fraccoredeviation}
		\sum_{c \in A_i \setminus W} z_c +\sum_{c \in A_i \cap W} \min\{y_i,z_c\} \ge \left(u_i(W)+1\right)y_i \qquad \text{for every } i \in N.
	\end{equation}
	A committee $W \in \mathcal{W}$ satisfies \emph{core+} if it admits no fractional blocking coalition.
\end{definition}

To relate this definition to the ordinary core, suppose first that $y \in \{0,1\}^n$ and $z \in \{0,1\}^m$ describe a coalition and its proposal, respectively.\footnote{All inequalities in \Cref{def:core+-deviation} are homogeneous under positive scaling, meaning that every feasible solution $(y,z) \in \mathbb{R}_{\ge 0}^n \times \mathbb{R}_{\ge 0}^m$ can be scaled to a feasible solution in $[0,1]^n \times [0,1]^m$.} For each $i \in S$, defining $S = \{i \in N : y_i=1\}$, the left-hand side of \eqref{eq:fraccoredeviation} equals $u_i(T)$, so the condition becomes exactly \eqref{eq:coredeviation}. For $i \notin S$, this constraint automatically holds. Together with the affordability constraint, this recovers the ordinary core.

Second, demanding only $z \in [0,1]^m$ instead of $z \in \{0,1\}^m$ is equivalent to allowing coalitions to propose fractional committees where a voter receives utility $z_c$ from $c \in A_i$. 

Finally, relaxing $y \in \{0,1\}^n$ to $y \in [0,1]^n$ enables voters to deviate partially.
The $y_i$-fraction of the voter receives utility from $z$, while the $(1-y_i)$-fraction remains with $W$. Consequently, the latter fraction receives utility $1-y_i$ from each candidate in $A_i \cap W$. The $y_i$-fraction obtains utility $z_c$ from each $c \in A_i \setminus W$ but can receive at most the ``remaining'' utility $1-(1-y_i)=y_i$ from candidates in $A_i \cap W$.\footnote{Taking $\min\{y_i,z_c\}$ for all $c \in A_i$ results in a characterization of \emph{Lindahl priceability} (see \Cref{prop:lindahl-eq-def} and \citet{munagala2022auditing}).} A deviation for voter $i$ with fraction $y_i$ is beneficial if she gains at least $y_i$ in total utility compared to $W$, which corresponds exactly to \eqref{eq:fraccoredeviation}.

\medskip

The existence of a fractional blocking coalition can be tested via a linear feasibility problem (by encoding the minimum using auxiliary variables, see \Cref{app:core+proofs}). For purposes of designing our voting rule satisfying core+, however, it is convenient to work with an equivalent definition of core+ in terms of \emph{payment systems}. 
For a given committee $W$, a \emph{payment system} $(r,p)$ specifies, for each voter $i \in N$, a \emph{reserve} $r_i \in [0,1]$ which denotes her unspent budget, and a \emph{payment} $p_{ic} \in [0,1]$ for each approved winner $c \in A_i \cap W$. Note that voters make no payments to committee members they do not approve.  Let $p_i = (p_{ic})_{c \in A_i \cap W}$ be voter $i$'s payment vector. We impose the following conditions on the reserves and payments:
\begin{enumerate}
	\renewcommand{\theenumi}{(P\arabic{enumi})}\renewcommand{\labelenumi}{\theenumi}\setcounter{enumi}{0}
	\item \label{cond:p1} For each $i \in N$, $r_i + \sum_{c \in A_i \cap W} p_{ic} = 1$.
	\item \label{cond:p2} For each $i \in N$, $r_i \ge p_{ic}$ for all $c \in A_i \cap W$.
\end{enumerate}
Condition \ref{cond:p1} gives each voter a unit budget, divided between payments and reserve. Condition \ref{cond:p2} prevents a voter from committing more to any committee member than she retains for possible ``outside alternatives''. Hence, every payment is backed by an equally large reserve, linking support for the current committee to the budget available for a deviation.
Note that \ref{cond:p1} and \ref{cond:p2} are trivially satisfied by the payment system with $r_i=1$ for all $i \in N$. 

To encode fairness into the payment systems, we require two more conditions on the total payments to individual candidates:
\begin{enumerate}
	\renewcommand{\theenumi}{(P\arabic{enumi})}\renewcommand{\labelenumi}{\theenumi}\setcounter{enumi}{2}
	\item \label{cond:p3} For each $c \in C \setminus W$, $\sum_{i \in N_c} r_i < q$.
	\item \label{cond:p4} For each $c \in W$, $\sum_{i \in N_c} p_{ic} \le q$.
\end{enumerate}
Condition \ref{cond:p3} ensures that the supporters of any unselected candidate $c$ collectively hold less than the quota $q$ in reserve; intuitively they otherwise would have a justified complaint that $c$ was not elected. Condition \ref{cond:p4} requires that no selected candidate receives more than $q$ in total payments; intuitively otherwise its supporters would have a justified complaint that they were overcharged.

We write $\mathcal P_{1234}(W)$ for the set of payment systems satisfying \ref{cond:p1}--\ref{cond:p4}. More generally, for $X\subseteq[4]$, we write $\mathcal P_X(W)$ for the set of payment systems satisfying \textup{(P$x$)} for every $x\in X$.

\begin{restatable}{theorem}{corepaymentequiv} \label{thm:p1234-core}
	A committee $W \in \mathcal W$ satisfies core+ if and only if $\mathcal P_{1234}(W)$ is non-empty.
\end{restatable}

The characterization yields a polynomial-time algorithm for checking core+ membership via linear programming. We defer the proof to \Cref{app:core+proofs} and explain here why a payment system $(r, p) \in \mathcal P_{1234}(W)$ guarantees that $W$ is in the core. 

Suppose that $W$ is not in the core and there exists a proposal $T \subseteq C$ by coalition $S$ with $\lvert S \rvert \ge q\lvert T \rvert$ that all voters in $S$ strictly prefer.
Since every voter $i \in S$ strictly improves, each approved winner omitted in $T$ can be matched to a distinct approved candidate in $T \setminus W$, with at least one additional approved candidate in $T\setminus W$ remaining unmatched.
For each $i \in S$, we create hypothetical payments where every candidate $c \in T \cap W$ receives their original payment $p_{ic}$ and outside candidates receive $r_i$, which is weakly larger than any other payment by \ref{cond:p2}. Hence, the payment of every voter $i$ is at least their unit budget, defined by \ref{cond:p1}. The \emph{total payment} to $T$ is therefore at least $\lvert S \rvert$.
Conversely, each candidate $c \in T\cap W$ receives at most $\sum_{i \in N_c}p_{ic} \le q$ by \ref{cond:p4} and each candidate $c'$ outside $W$ receives $\sum_{i \in N_{c'}}r_i < q$ by \ref{cond:p3}. The \emph{total payment} to candidates in $T$ is therefore strictly less than $q \lvert T \rvert$.
Combining these two bounds on the total payment gives $\lvert S \rvert < q \lvert T \rvert$, and hence the proposal $T$ by coalition $S$ was not affordable in the first place.

\begin{figure}[t]
    \centering
    \definecolor{waterkeep}{RGB}{45,105,145}
    \definecolor{waternew}{RGB}{230,145,45}

    \begin{tikzpicture}[
    	scale=0.85,
        axiom/.style={
            draw=waterkeep!70!black,
            fill=waterkeep!6,
            rounded corners,
            minimum width=18mm,
            minimum height=7mm,
            align=center,
            inner sep=2pt,
            font=\small
        },
        highlight/.style={
            axiom,
            draw=waternew!75!black,
            fill=waternew!15,
            thick
        },
        implication/.style={
            -{Latex[length=1.8mm]},
            draw=black!65
        },
        newimplication/.style={
            -{Latex[length=1.8mm]},
            draw=waternew!75!black,
            line width=1pt
        },
        newequivalence/.style={
            {Latex[length=1.8mm]}-{Latex[length=1.8mm]},
            draw=waternew!75!black,
            line width=1pt
        }
    ]
        \node[axiom, densely dashed, minimum width=35mm]
            (stabprice) at (-0.4,4.2)
            {stable priceability};

        \node[highlight] (coreplus) at (-4,2.8)
            {core+};

        \node[axiom, minimum width=47mm]
            (frugalLindahl) at (4.5,2.8)
            {frugal Lindahl priceability};

        \node[axiom] (fjrplus) at (-4,0)
            {$\mathrm{FJR+}$};

        \node[axiom, minimum width=35mm]
            (Lindahlprice) at (-0.4,1.4)
            {Lindahl priceability};

        \node[axiom] (core) at (-0.4,0)
            {core};

        \node[axiom] (fjr) at (-0.4,-1.4)
            {$\mathrm{FJR}$};

        \node[axiom] (ejrplus) at (-4,-1.4)
            {$\mathrm{EJR+}$};

        \node[axiom] (ejr) at (-0.4,-2.8)
            {$\mathrm{EJR}$};

        \draw[implication] (stabprice) -- (coreplus);
        \draw[implication] (stabprice) -- (frugalLindahl);

        \draw[newequivalence] (coreplus) -- (frugalLindahl);
        \draw[newimplication] (coreplus) -- (fjrplus);

        \draw[implication] (coreplus) -- (Lindahlprice);
        \draw[implication] (frugalLindahl) -- (Lindahlprice);
        \draw[implication] (Lindahlprice) -- (core);
        \draw[implication] (core) -- (fjr);
        \draw[implication] (fjrplus) -- (ejrplus);
        \draw[implication] (fjrplus) -- (fjr);
        \draw[implication] (fjr) -- (ejr);
        \draw[implication] (ejrplus) -- (ejr);
        \path (-7.9,0) (7.1,0);
    \end{tikzpicture}
    \caption{
        Relations among strong fairness axioms in approval-based committee elections.
        Core+ is equivalent to \emph{frugal Lindahl priceability} \citep{golz2026apportionment} and implies $\mathrm{FJR+}$. Stable priceability \citep{peters2021market} is stronger than core+ but is not always satisfiable. Thick orange arrows indicate new implications.
    }
    \label{fig:core-hierarchy}
\end{figure}

To conclude this section, we show that satisfying only \ref{cond:p1}--\ref{cond:p3} already yields good fairness guarantees in terms of $\mathrm{EJR+}$ and $f$-representativeness \citep[see][for definitions of these axioms]{BrPe23a}, and does not allow disjoint core deviations (a property discussed by \citealp{BGP+24a}). Thus, in particular, core+ implies all these properties, which we define formally as part of the statement of the following proposition.

\begin{proposition} \label{thm:p123-properties}
	Every committee $W$ with $\mathcal P_{123}(W)\neq \emptyset$ satisfies the following properties:
	\begin{itemize}
		\item $\mathrm{EJR+}$: for every $c \in C \setminus W$, for all $\ell \in [k]$, and every group $S \subseteq N_c$ with $\lvert S \rvert \ge \ell q$, there is some $i \in S$ with $u_i(W) \ge \ell$.
		\item  $(\ell-1)$-representative for all $\ell \in [k]$: for every $c \in C \setminus W$ and every group $S \subseteq N_c$ with $\lvert S \rvert \ge \ell q$, the average utility of $S$ satisfies $\frac{1}{\lvert S \rvert} \sum_{i \in S} u_i(W) \ge \ell -1$.
		\item no disjoint core deviations: for every $T \subseteq C$ with $T \cap W = \emptyset$, there exists no non-empty coalition $S \subseteq N$ with $\lvert S \rvert \ge \lvert T \rvert q$ and $u_i(T) > u_i(W)$ for each $i \in S$.
	\end{itemize}
\end{proposition}
\begin{proof}
	Let $(r,p) \in \mathcal P_{123}(W)$. Observe that for every $i \in N$, using \ref{cond:p1} and \ref{cond:p2} we have $1 = r_i + \sum_{d\in A_i \cap W} p_{id} \le r_i + u_i(W) r_i$, and hence, $r_i \ge \frac{1}{u_i(W) + 1}$.
	
    $\mathrm{EJR+}$:
	Let $c \in C \setminus W$ and $\ell \in [k]$. Suppose $S \subseteq N_c$ is a group with $\lvert S \rvert \ge \ell q$ such that every $i \in S$ has $u_i(W) \le \ell - 1$.  From our observation above, $r_i \ge \frac{1}{\ell}$, and hence, $\sum_{i \in N_c} r_i \ge \lvert S \rvert \cdot \frac{1}{\ell} \ge q$, contradicting \ref{cond:p3}.
	
	Representativeness: 
	Let $c \in C \setminus W$ and $\ell \in [k]$. Suppose $S \subseteq N_c$ is a group with $\lvert S \rvert \ge \ell q$. We show that their average utility $\frac{1}{\lvert S \rvert} \sum_{i \in S} u_i(W)$ is strictly larger than $\ell - 1$. Note that by \ref{cond:p3}, $\sum_{i \in S} \frac{1}{u_i(W) + 1} \le \sum_{i \in N_c} \frac{1}{u_i(W) + 1} \le \sum_{i \in N_c} r_i < q \le \frac{\lvert S \rvert}{\ell}$. Thus, the average of the values $\frac{1}{u_i(W) + 1}$ over $i \in S$ is strictly below $\frac{1}{\ell}$. Using the AM-HM inequality, so is their harmonic mean, which equals $\lvert S \rvert/(\sum_{i\in S} (u_i(W) + 1))$. Taking reciprocals, $\frac{1}{\lvert S \rvert}\sum_{i \in S} (u_i(W) + 1) > \ell$, and hence, $\frac{1}{\lvert S \rvert} \sum_{i \in S} u_i(W) > \ell - 1$, as desired.
	
	No disjoint core deviations:
	The argument in \Cref{sec:core+} that conditions \ref{cond:p1}--\ref{cond:p4} imply core membership uses \ref{cond:p4} only to show that the total payments to $T$ sum up to strictly less than $\lvert T \rvert q$. But under the assumption that $T \cap W = \emptyset$, this conclusion can be deduced using only \ref{cond:p3}, since the total payments to $T$ will then be $\sum_{c \in T} \sum_{i \in N_c} r_i < \sum_{c \in T} q = \lvert T \rvert q$.
\end{proof}

\section{Proportional Approval Voting} \label{section:pav}

Proportional approval voting (\pav) selects committees $W \in \mathcal W_k$ that maximize the sum of voters' harmonic utilities. Formally, the \pav score of a committee $W$ is defined as 
\[
	\pav(W) = \sum_{i \in N}H_{u_i(W)}, \quad \text{with}\quad H_t = \sum_{j=1}^t \frac{1}{j},
\]
where $H_0 = 0$. Thus, each additional approved committee member contributes less to a voter's score. The first contributes $1$, the second $\frac{1}{2}$, the third $\frac{1}{3}$, and so on. 

We here show that the harmonic structure of \pav admits a natural interpretation in terms of voter payments for which the reserve $r_i$ is given by $\frac{1}{u_i(W)+1}$. Given a committee $W$, consider its \emph{uniform payment system}, in which every voter divides their budget equally among their approved committee members and their unspent reserve. Formally, for each voter $i \in N$, we set 
\begin{equation} \label{eq:uniform-payment-system}
\begin{aligned} 
	r_i &= \frac{1}{u_i(W) + 1} \text{ for all $i \in N$, and} \\
	p_{ic} &= \frac{1}{u_i(W) + 1} \text{ for all $i \in N$ and $c \in A_i \cap W$.}
\end{aligned}
\end{equation}
Thus, under the uniform payment system, each voter splits their unit budget equally between the reserve and the $\lvert A_i \cap W \rvert$ approved committee members. 
This payment system clearly satisfies \ref{cond:p1} and \ref{cond:p2}. In \Cref{thm:pav-p123}, we will show that \pav-maximizing committees additionally satisfy \ref{cond:p3}.
This proof additionally serves as a blueprint for proving core stability of committees maximizing the \emph{harmonic entropy objective}. In particular, the proof uses the following structure: (a) if \ref{cond:p3} were violated by a witness $c \in C \setminus W$, then adding $c$ to the committee would increase its objective value by at least $q$, and (b) after having added $c$, we can remove a candidate whose removal decreases the objective value by strictly less than $q$. Combining the two yields a single-candidate swap that contradicts the optimality of $W$.

\begin{theorem}
	\label{thm:pav-p123}
	Let $W$ be a committee and $(r,p)$ its uniform payment system.
	\begin{enumerate}
		\item[(a)] For every candidate $c \in C \setminus W$, 
		\[
			\pav(W \cup \{c\}) - \pav(W) = \sum_{i \in N_c} r_i.
		\]
		\item[(b)] For every $W' \subseteq C$ with $\lvert W' \rvert = k+1$, there exists a candidate $d \in W'$ with
		\[
			\pav(W') - \pav(W' \setminus \{d\}) \leq \frac{n}{k+1}.
		\]
	\end{enumerate}
	Hence, if $W \in \mathcal W_k$ maximizes $\pav(W)$, then $\sum_{i \in N_c} r_i < q$ for every $c \in C \setminus W$.
	Thus, the uniform payment system of every \pav-optimal committee satisfies \ref{cond:p1}--\ref{cond:p3}.
\end{theorem}
\begin{proof}
	(a) For a voter $i \in N_c$, their contribution to the \pav score of $W$ is $1 + \frac{1}{2} + \dots + \frac{1}{u_i(W)}$. After adding $c \in C \setminus W$ to $W$, this contribution goes up by $\frac{1}{u_i(W) + 1}$, which is equal to $r_i$ by the definition of the uniform payment system.
	
	(b) Write $N' = \{ i \in N : A_i \cap W' \neq \emptyset\}$. For each $d \in W'$, the reduction in the \pav score of $W'$ that would result from deleting $d$ from $W'$ is $\sum_{i \in N_d} \frac{1}{\lvert A_i \cap W' \rvert}$. Summing these reductions over all $k+1$ members of $W'$ we get
	\[
		\sum_{d \in W'} \sum_{i \in N_d} \frac{1}{\lvert A_i \cap W' \rvert}
		= \sum_{i \in N'} \sum_{d \in A_i \cap W'} \frac{1}{\lvert A_i \cap W' \rvert}
		= \sum_{i \in N'} 1
		\le n.
	\]
	Hence, the average reduction is at most $\frac{n}{k+1}$, and hence, there exists a candidate $d \in W'$ whose removal from $W'$ would decrease the \pav score of $W'$ by at most $\frac{n}{k+1}$, as required.
		
	Now combine (a) and (b). Suppose there did exist a loser $c \in C \setminus W$ with $\sum_{i \in N_c} r_i \ge q$. Then by (a), the \pav score of the committee $W' = W \cup \{c\}$ would be at least $\pav(W) + q$. By (b), we can find some $d$ in $W'$ whose removal yields a committee $W'' = W' \setminus \{d\}$ with \pav score at least $(\pav(W) + q) - \frac{n}{k+1}$. Since $q > \frac{n}{k+1}$, the committee $W''$ has strictly higher \pav score than $W$, which is a contradiction to $W$ being selected by \pav. We have proved that the uniform payment system $(r,p)$ satisfies \ref{cond:p3}, and by definition it satisfies \ref{cond:p1} and \ref{cond:p2}, establishing the theorem.
\end{proof}

\begin{wrapstuff}[
    r,
    top=2,
    width=0.3\textwidth,
    hsep=1.2em,
    vsep=0.4em
]
    \begin{minipage}{\linewidth}
    \centering
    \definecolor{waterkeep}{RGB}{45,105,145}
    \definecolor{waternew}{RGB}{230,145,45}

    \begin{tikzpicture}[
        x=0.85cm,
        y=0.48cm,
        winner/.style={
            draw=waterkeep!70!black,
            thick
        },
        excluded/.style={
            draw=waternew!75!black,
            dashed,
            thick
        }
    ]
        \foreach \i in {1,...,4}
            \node[anchor=north] at (\i-0.5,-0.15)
                {$v_{\i}$};

        \foreach \y/\c in {0/1,1/2}{
            \draw[winner,fill=waternew!45]
                (0,\y) rectangle (2,\y+1);
            \node at (1,\y+0.5) {$c_{\c}$};
        }

        \draw[excluded,fill=waternew!10]
            (0,2) rectangle (1,3);
        \node at (0.5,2.5) {$c_3$};

        \draw[excluded,fill=waternew!10]
            (1,2) rectangle (2,3);
        \node at (1.5,2.5) {$c_4$};

        \foreach \y/\c in {
            0/5,1/6,2/7,3/8,4/9,5/10
        }{
            \draw[winner,fill=waterkeep!40]
                (2,\y) rectangle (4,\y+1);
            \node at (3,\y+0.5) {$c_{\c}$};
        }
    \end{tikzpicture}
    \captionsetup{type=figure}
    \captionof{figure}{Core violation of a \pav committee.}
    \label{fig:pav-p4-core-counterexample}
    \end{minipage}
\end{wrapstuff}

However, a \pav committee can still violate the core and \ref{cond:p4}. Consider the instance in \Cref{fig:pav-p4-core-counterexample} with four voters, committee size $k=8$, and Hare quota $q=\frac{n}{k}=\frac{1}{2}$. Each rectangle represents a candidate. Voters approve those candidates whose rectangles are drawn above the voter. 
The committee $W=\{c_1,c_2,c_5,\ldots,c_{10}\}$ (marked in solid) maximizes the \pav objective (tied with some other committees).
The committee $W$ is not in the core: The coalition $S=\{v_1,v_2\}$ can afford $T=\{c_1,c_2,c_3,c_4\}$, since $\lvert S \rvert=2=q\lvert T \rvert$. Both voters strictly prefer $T$ to $W$, as their utilities increase from $2$ to $3$.

Given the core failure, it is unsurprising that under the uniform payment system in \eqref{eq:uniform-payment-system}, voters $v_1$ and $v_2$ each pay $\frac{1}{3}$ to both $c_1$ and $c_2$. Hence, both candidates receive $\frac{2}{3}$ exceeding $q=\frac{1}{2}$ and thereby violating \ref{cond:p4}. Intuitively, \pav overcharges voters $v_1$ and $v_2$.
This example thereby illustrates the limitations of uniform payments. 
Requiring the payments to satisfy \ref{cond:p4} forces $r_1 \ge q$ or $r_2 \ge q$, violating \ref{cond:p3} for either $c_3$ or $c_4$, revealing the demand for larger representation.
In particular, \pav does not distinguish between committees $W$ and $W' = \{c_1,c_2,c_3,c_5, \dots, c_9\}$. The latter satisfies conditions \ref{cond:p1}--\ref{cond:p4}, e.g., via the payment system where $(r'_1,p'_1)=(\frac{3}{10},\frac{1}{5},\frac{1}{5},\frac{3}{10})$, $(r'_2,p'_2)=(\frac{2}{5},\frac{3}{10},\frac{3}{10})$, and $v_3$ and $v_4$ each split their budget uniformly among their reserve and $c_5$--$c_9$. 
This suggests evaluating committees using both voters' numbers of approved winners and their payment distributions subject to the quota constraints \ref{cond:p4}. 
Our harmonic entropy objective captures this additional flexibility.

\section{Harmonic Entropy} \label{section:harmonic-entropy}

We now develop \emph{harmonic entropy}, a new measure of how evenly mass is distributed across the coordinates of a probability vector. This perspective arises naturally from our payment framework: By \ref{cond:p1} and \ref{cond:p2}, a voter's payments and reserve form a probability distribution, with the reserve being the largest coordinate. We seek an objective that rewards balanced distributions, recovering the harmonic numbers of \pav at uniform distributions, while allowing non-uniform payments to accommodate the quota constraint \ref{cond:p4}.

Our construction is guided by the swap argument in \Cref{thm:pav-p123}. Adding an approved candidate increases a voter's \pav score by exactly the reserve $r_i$. We aim to preserve this marginal-gain identity when extending to non-uniform payment distributions. In particular, if the supporters of an unelected candidate collectively hold at least $q$ in reserve, then adding that candidate should increase the \emph{harmonic entropy objective} by at least $q$. This will allow us to establish \ref{cond:p3} through an argument similar to the one used for \pav.

To achieve this, we introduce the \emph{Robin Hood} operator, which adds one or more coordinates to a probability distribution by redistributing mass from its largest coordinates. We will then define \emph{harmonic entropy} so that the increase in the objective under a single operation equals the largest coordinate of the original distribution, corresponding to a voter's reserve.

\subsection{Robin Hood Operator} \label{section:robin-hood}

Given a probability distribution $x = (x_1, \dots, x_d) \in \Delta^d$ and a number of steps $\ell \ge 1$, we define the Robin Hood operator $\Phi^\ell \colon \Delta^d \to \Delta^{d+\ell}$ by
\begingroup
\setlength{\belowdisplayskip}{4pt}
\setlength{\belowdisplayshortskip}{4pt}
\[
	\Phi^\ell(x) = \Big(\min\{x_1,t_\ell\}, \dots, \min\{x_d,t_\ell\}, \underbrace{t_\ell,\ldots,t_\ell}_{\ell\text{ times}} \Big),
\]
\endgroup
where $t_\ell \in [0,1]$ is the unique solution to $\ell \cdot t_\ell = \sum_{j=1}^d (x_j - t_\ell)_+$.\footnote{Note that the left-hand side is strictly increasing in $[0,1]$, while the right-hand side $\sum_{j \in [d]} (x_j - t)_+$ is strictly decreasing from $1$ (as $t = 0$) until it reaches $0$ (at some point $t \le 1$). By the intermediate value theorem, there exists a point where we have equality; due to strict monotonicity, the solution must be unique.} We set $\Phi^0(x) = x$ and write $\Phi = \Phi^1$.
Intuitively, the operator $\Phi^\ell$ lowers all coordinates exceeding the \emph{water level} $t_\ell$ to this common value and assigns the removed mass equally among $\ell$ new coordinates. Note that this operation preserves the total mass of the original distribution and no existing coordinate increases, i.e., $\Phi^\ell(x)_j \leq x_j$ for all $j \in [d]$. Moreover, $\Phi^\ell$ agrees with repeated applications of the single-coordinate operator, that is, for every $\ell \ge 0$, we have $\Phi^{\ell+1}(x) = \Phi(\Phi^\ell(x))$ (see \Cref{appendix:robin-hood} for a proof). \Cref{fig:robin-hood-progression} illustrates three successive applications to $x=(\frac{1}{2},\frac{1}{3},\frac{1}{6})$.

\begin{figure}[ht!]
\centering
\begin{tikzpicture}[
    x=1cm,
    y=5cm,
    >=Stealth,
    font=\small
]
    \definecolor{waterkeep}{RGB}{45,105,145}
    \definecolor{waternew}{RGB}{230,145,45}

    \def\bw{0.34}
    \def\st{0.44}

    \newcommand{\barfull}[4]{
        \fill[#4]
            ({#1+(#2)*\st},0)
            rectangle
            ({#1+(#2)*\st+\bw},#3);
        \draw[black!55]
            ({#1+(#2)*\st},0)
            rectangle
            ({#1+(#2)*\st+\bw},#3);
    }

    \newcommand{\barsplit}[5]{
        \fill[#5]
            ({#1+(#2)*\st},0)
            rectangle
            ({#1+(#2)*\st+\bw},#4);

        \fill[#5,fill opacity=0.25]
            ({#1+(#2)*\st},#4)
            rectangle
            ({#1+(#2)*\st+\bw},#3);

        \draw[black!55]
            ({#1+(#2)*\st},0)
            rectangle
            ({#1+(#2)*\st+\bw},#3);
    }

    \barsplit{0}{0}{1/2}{5/18}{waterkeep!75}
    \barsplit{0}{1}{1/3}{5/18}{waterkeep!75}
    \barfull{0}{2}{1/6}{waterkeep!75}

    \draw[dashed,thick,black!75]
        (-0.04,5/18)--(1.26,5/18);

    \barsplit{3.25}{0}{5/18}{5/24}{waterkeep!75}
    \barsplit{3.25}{1}{5/18}{5/24}{waterkeep!75}
    \barfull{3.25}{2}{1/6}{waterkeep!75}
    \barsplit{3.25}{3}{5/18}{5/24}{waternew!85}

    \draw[dashed,thick,black!75]
        (3.21,5/24)--(4.95,5/24);

    \barsplit{7.08}{0}{5/24}{1/6}{waterkeep!75}
    \barsplit{7.08}{1}{5/24}{1/6}{waterkeep!75}
    \barfull{7.08}{2}{1/6}{waterkeep!75}
    \barsplit{7.08}{3}{5/24}{1/6}{waterkeep!75}
    \barsplit{7.08}{4}{5/24}{1/6}{waternew!85}

    \draw[dashed,thick,black!75]
        (7.04,1/6)--(9.22,1/6);

    \foreach \j in {0,...,4}{
        \barsplit{11.4}{\j}{1/6}{1/7}{waterkeep!75}
    }

    \barsplit{11.4}{5}{1/6}{1/7}{waternew!85}

    \draw[dashed,thick,black!75]
        (11.36,1/7)--(13.98,1/7);

    \foreach \x/\t/\value [count=\ell] in {
        1.32/{5/18}/{\frac{5}{18}},
        5.01/{5/24}/{\frac{5}{24}},
        9.28/{1/6}/{\frac{1}{6}},
        14.04/{1/7}/{\frac{1}{7}}
    }{
        \node[anchor=west,font=\footnotesize,inner sep=1pt]
            at (\x,\t) {$t_\ell=\value$};
    }

    \newcommand{\distribution}[4]{
        \node[anchor=base,font=\footnotesize]
            at ({#1+#4*\st/2+\bw/2},-0.19) {$#2$};
        \foreach \value [count=\j from 0] in {#3}{
            \node[font=\footnotesize,inner sep=0pt] (entry-#4-\j)
                at ({#1+\j*\st+\bw/2},-0.075)
                {$\value\ifnum\j<#4\mathrlap{,}\fi$};
        }
        \node[anchor=east,font=\footnotesize,inner sep=0pt] (prefix-#4)
            at ($(entry-#4-0.west)+(0.015,0)$) {$($};
        \node[anchor=west,font=\footnotesize,inner sep=0pt] (suffix-#4)
            at ($(entry-#4-#4.east)-(0.015,0)$) {$)$};
    }
    \distribution{0}{x}{\frac12,\frac13,\frac16}{2}
    \distribution{3.25}{\Phi(x)}{\frac5{18},\frac5{18},\frac16,\frac5{18}}{3}
    \distribution{7.08}{\Phi^2(x)}{\frac5{24},\frac5{24},\frac16,\frac5{24},\frac5{24}}{4}
    \distribution{11.4}{\Phi^3(x)}{\frac16,\frac16,\frac16,\frac16,\frac16,\frac16}{5}

    \foreach \xa/\xb/\y in {1.44/3.07/0.37,5.13/6.90/0.285,9.40/11.22/0.245}{
        \draw[->,thick] (\xa,{\y+0.02})
            to[bend left=10] node[midway,above,font=\footnotesize] {$\Phi$}
            (\xb,{\y-0.02});
    }

\end{tikzpicture}

\caption{Successive applications of the Robin Hood operator to $x=(\frac{1}{2},\frac{1}{3},\frac{1}{6})$. Dashed lines indicate the water level for the next application, and the transparent portions indicate the mass to be redistributed. Orange marks the most recently introduced coordinate.}
\label{fig:robin-hood-progression}
\end{figure}

Note that repeatedly applying $\Phi$ to any strictly positive distribution $x$ eventually yields a uniform distribution. To see this, let $M$ be an integer strictly larger than $\frac{1}{x_j}$ for every $j \in [d]$. After $M$ applications of $\Phi$ to $x$, we obtain a new distribution $\Phi^M(x) \in \Delta^{d+M}$. In the process, we add $M$ new coordinates, all of which have the same value and attain the maximum. Because all coordinates sum to $1$, the value in the new coordinates is at most $\frac{1}{M}$, i.e., strictly smaller than all original values $x_j$. Hence, every original coordinate must have been lowered at some point in the process. But note that, whenever $\Phi$ changes a coordinate, it becomes a largest coordinate, and a largest coordinate always remains a largest coordinate. So in fact, \emph{every} coordinate of $\Phi^M(x)$ is largest, and thus, the distribution is uniform.

Building on this, we can see that $\Phi$ preserves uniformity: $\Phi(\frac{1}{d},\dots,\frac{1}{d}) = (\frac{1}{d+1},\dots,\frac{1}{d+1})$. This mirrors what happens under the uniform payment system of \pav when a candidate is added, and each voter's budget is again spread equally among all coordinates. The Robin Hood operator extends this balancing principle to non-uniform distributions by spreading budget as equally as possible among coordinates, subject to not increasing any existing one.\footnote{One can make this statement precise using the concept of \emph{majorization}, a standard ordering of probability distributions that compares their degree of concentration. We show that $\Phi^\ell(x)$ is majorized by every distribution $q \in \Delta^{d+\ell}$ satisfying $q_j \leq x_j$ for all $j \in [d]$. Thus, our operator redistributes mass as evenly as possible without increasing any original coordinate. We provide the proof in \Cref{appendix:robin-hood}.}

Lastly, for our analysis of harmonic entropy, we need to understand how the water level depends on the original distribution. Given $x \in \Delta^d$ and $\ell\geq0$, write \[t_\ell(x) \coloneqq \lVert\Phi^\ell(x)\rVert_\infty.
\]
The following lemma shows that the water level $t_\ell(x)$ can either be expressed as a maximum of linear functions or as a constrained minimization problem. We defer the proof to \Cref{appendix:robin-hood}.

\begin{restatable}{lemma}{waterlevelCharacterization} \label{lem:water-level-characterization}
    For every $x\in\Delta^d$ and integer $\ell\geq0$,
	\begin{align}
		t_\ell(x) &=\max_{\emptyset\neq S\subseteq[d]} \frac{\sum_{j\in S}x_j}{\lvert S \rvert+\ell}, \label{eq:water-level-maximum}\\
		t_\ell(x) &=\min\Big\{t\geq0 : \sum_{j=1}^d(x_j-t)_+ \leq \ell \cdot t \Big\}. \label{eq:water-level-minimum}
	\end{align}
\end{restatable}

\subsection{Definition of Harmonic Entropy}

Having established how the Robin Hood operator redistributes mass when new coordinates are introduced, we can finally formalize the idea of harmonic entropy. For every strictly positive distribution $x \in \Delta^d$, we define \emph{harmonic entropy} recursively by
\[
	F(x) = 
	\begin{cases}
		H_{d-1} 						& \text{if $x = (\frac1d, \dots, \frac1d)$,} \\
		F(\Phi(x)) - \lVert x \rVert_\infty 	& \text{otherwise}.
	\end{cases}
\]  

\begin{wrapstuff}[
	r,
	top=8,
	width=0.38\textwidth,
	hsep=1.0em,
	vsep=0.3em,
	type=figure
	]
	\centering
	
	\begin{tikzpicture}[
		x=5cm, y=5cm, font=\footnotesize,
		line join=round, line cap=round
		]
		\definecolor{simplexEdge}{HTML}{24313B}
		\definecolor{simplexBand0}{HTML}{481F70}
		\definecolor{simplexBand1}{HTML}{3B528B}
		\definecolor{simplexBand2}{HTML}{287C8E}
		\definecolor{simplexBand3}{HTML}{20A486}
		\definecolor{simplexBand4}{HTML}{5EC962}
		\definecolor{simplexBand5}{HTML}{B2DD2D}
		\definecolor{simplexBand6}{HTML}{EAE51A}
		\coordinate (e1) at (0,0);
		\coordinate (e2) at (1,0);
		\coordinate (e3) at (0.5,0.866025);
		\coordinate (u) at (0.5,0.288675);
		\fill[simplexBand0] (e1) -- (e2) -- (e3) -- cycle;
		\path[fill=simplexBand1, draw=white, draw opacity=0.65, line width=0.35pt]
		(0.52827,0.81705) -- (0.52257,0.82043) -- (0.51601,0.82290) -- (0.50838,0.82457) --
		(0.50086,0.82517) -- (0.49162,0.82457) -- (0.48399,0.82290) -- (0.47743,0.82043) --
		(0.47173,0.81705) -- (0.02827,0.04897) -- (0.02806,0.04426) -- (0.02909,0.03657) --
		(0.03084,0.03033) -- (0.03393,0.02309) -- (0.03730,0.01732) -- (0.04201,0.01118) --
		(0.04749,0.00577) -- (0.05273,0.00192) -- (0.05655,0.00000) -- (0.94345,0.00000) --
		(0.94923,0.00325) -- (0.95465,0.00770) -- (0.95991,0.01347) -- (0.96419,0.01969) --
		(0.96829,0.02799) -- (0.97066,0.03543) -- (0.97180,0.04234) -- (0.97173,0.04897) --
		(0.52827,0.81705) -- cycle;
		\path[fill=simplexBand2, draw=white, draw opacity=0.65, line width=0.35pt]
		(0.07299,0.12643) -- (0.07169,0.12032) -- (0.07063,0.10695) -- (0.07105,0.09612) --
		(0.07226,0.08666) -- (0.07432,0.07698) -- (0.07705,0.06802) -- (0.08070,0.05894) --
		(0.08595,0.04879) -- (0.09139,0.04041) -- (0.09743,0.03272) -- (0.10383,0.02587) --
		(0.11118,0.01925) -- (0.11876,0.01347) -- (0.12794,0.00770) -- (0.14004,0.00192) --
		(0.14599,0.00000) -- (0.85401,0.00000) -- (0.85996,0.00192) -- (0.86842,0.00577) --
		(0.88273,0.01452) -- (0.89617,0.02587) -- (0.90815,0.03977) -- (0.91685,0.05389) --
		(0.92339,0.06928) -- (0.92741,0.08468) -- (0.92932,0.10318) -- (0.92888,0.11548) --
		(0.92701,0.12643) -- (0.57299,0.73960) -- (0.56835,0.74378) -- (0.55801,0.75095) --
		(0.54179,0.75900) -- (0.53069,0.76284) -- (0.52170,0.76504) -- (0.50292,0.76716) --
		(0.49276,0.76695) -- (0.48233,0.76576) -- (0.47196,0.76357) -- (0.46296,0.76082) --
		(0.45228,0.75644) -- (0.44270,0.75138) -- (0.43165,0.74378) -- (0.42701,0.73960) --
		(0.07299,0.12643) -- cycle;
		\path[fill=simplexBand3, draw=white, draw opacity=0.65, line width=0.35pt]
		(0.13889,0.24056) -- (0.13347,0.22348) -- (0.13060,0.21080) -- (0.12784,0.18679) --
		(0.12699,0.16991) -- (0.12724,0.16265) -- (0.13010,0.13664) -- (0.13405,0.12056) --
		(0.14085,0.10154) -- (0.14828,0.08660) -- (0.15740,0.07247) -- (0.16954,0.05774) --
		(0.18338,0.04435) -- (0.20448,0.02887) -- (0.21064,0.02502) -- (0.22568,0.01732) --
		(0.24786,0.00770) -- (0.26027,0.00385) -- (0.27778,0.00000) -- (0.72222,0.00000) --
		(0.73200,0.00192) -- (0.74762,0.00606) -- (0.75700,0.00962) -- (0.77432,0.01732) --
		(0.78985,0.02529) -- (0.80114,0.03272) -- (0.81412,0.04234) -- (0.82857,0.05581) --
		(0.83394,0.06158) -- (0.84383,0.07420) -- (0.85086,0.08511) -- (0.85766,0.09815) --
		(0.86523,0.11796) -- (0.86932,0.13397) -- (0.87082,0.14292) -- (0.87303,0.16743) --
		(0.87134,0.19590) -- (0.86856,0.21554) -- (0.86111,0.24056) -- (0.63889,0.62546) --
		(0.62680,0.63870) -- (0.61726,0.64752) -- (0.60154,0.65935) -- (0.58151,0.67235) --
		(0.56807,0.67884) -- (0.55328,0.68504) -- (0.53224,0.69093) -- (0.51751,0.69328) --
		(0.49914,0.69431) -- (0.48617,0.69368) -- (0.47522,0.69223) -- (0.46776,0.69093) --
		(0.45202,0.68670) -- (0.43492,0.68017) -- (0.41682,0.67138) -- (0.39080,0.65380) --
		(0.37834,0.64375) -- (0.36111,0.62546) -- (0.13889,0.24056) -- cycle;
		\path[fill=simplexBand4, draw=white, draw opacity=0.65, line width=0.35pt]
		(0.25000,0.43301) -- (0.23638,0.39788) -- (0.22486,0.36252) -- (0.21234,0.31775) --
		(0.20279,0.27713) -- (0.20182,0.25718) -- (0.20182,0.23409) -- (0.20425,0.19630) --
		(0.21057,0.17227) -- (0.21824,0.15011) -- (0.22314,0.14015) -- (0.23178,0.12509) --
		(0.24039,0.11230) -- (0.25447,0.09623) -- (0.27212,0.07873) -- (0.29742,0.06158) --
		(0.32014,0.04811) -- (0.34139,0.03706) -- (0.37457,0.02694) -- (0.42638,0.01347) --
		(0.47337,0.00385) -- (0.50000,0.00000) -- (0.52663,0.00385) -- (0.54707,0.00770) --
		(0.59717,0.01925) -- (0.65104,0.03464) -- (0.66040,0.03780) -- (0.68326,0.05004) --
		(0.70858,0.06543) -- (0.72807,0.07890) -- (0.74553,0.09623) -- (0.75961,0.11230) --
		(0.76706,0.12317) -- (0.77578,0.13818) -- (0.78176,0.15011) -- (0.78462,0.15751) --
		(0.78943,0.17227) -- (0.79575,0.19630) -- (0.79818,0.23409) -- (0.79826,0.25319) --
		(0.79746,0.27520) -- (0.78938,0.31091) -- (0.77979,0.34677) -- (0.76687,0.38840) --
		(0.75576,0.41918) -- (0.75000,0.43301) -- (0.74091,0.44492) -- (0.71349,0.47700) --
		(0.68475,0.50754) -- (0.66395,0.52817) -- (0.63861,0.55184) -- (0.60837,0.57055) --
		(0.58590,0.58252) -- (0.56788,0.59099) -- (0.54390,0.59753) -- (0.52088,0.60197) --
		(0.50756,0.60275) -- (0.49020,0.60271) -- (0.47912,0.60197) -- (0.45610,0.59753) --
		(0.43212,0.59099) -- (0.39818,0.57420) -- (0.37818,0.56266) -- (0.36139,0.55184) --
		(0.33099,0.52325) -- (0.29848,0.49003) -- (0.27362,0.46237) -- (0.25000,0.43301) -- cycle;
		\path[fill=simplexBand5, draw=white, draw opacity=0.65, line width=0.35pt]
		(0.51556,0.48113) -- (0.47778,0.48113) -- (0.44889,0.47728) -- (0.44028,0.47391) --
		(0.40621,0.45723) -- (0.37994,0.44254) -- (0.34145,0.41820) -- (0.31523,0.39974) --
		(0.31257,0.39778) -- (0.31180,0.39645) -- (0.30855,0.36122) -- (0.30734,0.33988) --
		(0.30672,0.31571) -- (0.30713,0.28562) -- (0.30972,0.24778) -- (0.31111,0.23864) --
		(0.32222,0.21170) -- (0.34444,0.17321) -- (0.36222,0.15011) -- (0.36944,0.14434) --
		(0.39796,0.12509) -- (0.42678,0.10777) -- (0.44802,0.09623) -- (0.46710,0.08660) --
		(0.49923,0.07180) -- (0.50077,0.07180) -- (0.50380,0.07313) -- (0.53290,0.08660) --
		(0.57322,0.10777) -- (0.60204,0.12509) -- (0.63056,0.14434) -- (0.63778,0.15011) --
		(0.65556,0.17321) -- (0.67222,0.20207) -- (0.67778,0.21170) -- (0.68889,0.23864) --
		(0.69028,0.24778) -- (0.69287,0.28562) -- (0.69328,0.31571) -- (0.69266,0.33988) --
		(0.68999,0.37915) -- (0.68820,0.39645) -- (0.68743,0.39778) -- (0.65855,0.41820) --
		(0.62006,0.44254) -- (0.59379,0.45723) -- (0.55972,0.47391) -- (0.55111,0.47728) --
		(0.52222,0.48113) -- (0.51556,0.48113) -- cycle;
		\path[fill=simplexBand6, draw=white, draw opacity=0.65, line width=0.35pt]
		(0.46042,0.19702) -- (0.48333,0.17321) -- (0.49930,0.16044) -- (0.50070,0.16044) --
		(0.51667,0.17321) -- (0.55741,0.21554) -- (0.59204,0.27552) -- (0.60833,0.33198) --
		(0.61140,0.35218) -- (0.61070,0.35340) -- (0.59167,0.36084) -- (0.53222,0.37528) --
		(0.46537,0.37496) -- (0.40833,0.36084) -- (0.38930,0.35340) -- (0.38860,0.35218) --
		(0.39167,0.33198) -- (0.40889,0.27328) -- (0.44111,0.21747) -- (0.46042,0.19702) -- cycle;
		\draw[simplexEdge, line width=0.65pt] (e1) -- (e2) -- (e3) -- cycle;
		
		\node[fill=white, fill opacity=0.94, text opacity=1,
		rounded corners=2pt, inner sep=2.5pt, text=simplexEdge]
		at (0.5,0.405) {$F(x)=\frac32$};
		\draw[simplexEdge!70, line width=0.4pt] (0.5,0.357) -- (u);
		\filldraw[fill=white, draw=simplexEdge, line width=0.6pt] (u) circle (1.5pt);
		
		\node[below=3pt] at (e1) {$(1,0,0)$};
		\node[below=3pt] at (e2) {$(0,1,0)$};
		\node[above=3pt] at (e3) {$(0,0,1)$};
		
		\begin{scope}[shift={(0.14,-0.19)}, x=2.4cm, y=1cm]
			\fill[simplexBand0] (0.00000,0) rectangle (0.25000,0.10);
			\fill[simplexBand1] (0.25000,0) rectangle (0.50000,0.10);
			\fill[simplexBand2] (0.50000,0) rectangle (0.75000,0.10);
			\fill[simplexBand3] (0.75000,0) rectangle (1.00000,0.10);
			\fill[simplexBand4] (1.00000,0) rectangle (1.25000,0.10);
			\fill[simplexBand5] (1.25000,0) rectangle (1.40000,0.10);
			\fill[simplexBand6] (1.40000,0) rectangle (1.50000,0.10);
			\node[below=2pt, font=\scriptsize] at (0,0) {$0$};
			\node[below=2pt, font=\scriptsize] at (0.5,0) {$0.5$};
			\node[below=2pt, font=\scriptsize] at (1,0) {$1$};
			\node[below=2pt, font=\scriptsize] at (1.5,0) {$1.5$};
			\node[left=4pt] at (0,0.05) {$F$};
		\end{scope}
	\end{tikzpicture}
	\caption{Level sets of $F$ for three coordinates.}
	\label{fig:potential-simplex}
\end{wrapstuff}

This is well-defined, because every strictly positive distribution $x$ becomes uniform after finitely many applications of $\Phi$. For a distribution containing zeros, we define $F(x) \coloneqq F(x')$, where $x'$ is obtained by deleting zero coordinates.
As an example, consider \Cref{fig:robin-hood-progression}, which shows repeated applications of $\Phi$ to $x=(\frac{1}{2},\frac{1}{3},\frac{1}{6})$, from which we can deduce
\[
\begin{aligned}
F(\tfrac{1}{2},\tfrac{1}{3},\tfrac{1}{6})
&= F(\tfrac{5}{18},\tfrac{5}{18},\tfrac{1}{6},\tfrac{5}{18}) - \tfrac{1}{2} \\
&= F(\tfrac{5}{24},\tfrac{5}{24},\tfrac{1}{6},\tfrac{5}{24},\tfrac{5}{24}) - \tfrac{1}{2} - \tfrac{5}{18} \\
&= H_5 - \tfrac{1}{2} - \tfrac{5}{18} - \tfrac{5}{24} = \tfrac{467}{360} \approx 1.2972,
\end{aligned}
\]
since $\Phi^3(x)$ is uniform on six coordinates.

Harmonic entropy has two important properties that connect it to \pav. First, it recovers the harmonic scores on uniform distributions, i.e., for $x = (\frac{1}{d}, \dots, \frac{1}{d})$, we have $F(x) = H_{d-1}$, also motivating its name.
Second, harmonic entropy preserves the marginal-gain property underlying the \pav swap argument in \Cref{thm:pav-p123}. Recall that adding an approved candidate increases a voter's \pav score by exactly her reserve under the uniform payment system. Harmonic entropy extends this identity to non-uniform payment distributions. Introducing a new coordinate through the Robin Hood operator increases the objective by the largest coordinate of the original distribution, which corresponds to the voter's reserve by \ref{cond:p2}. Formally, $F$ satisfies
\[
	F(\Phi(x)) = F(x) + \lVert x \rVert_\infty \quad \text{for all $x \in \Delta^d$.}
\]
For non-uniform distributions, this follows directly from the definition. For uniform distributions, this follows from the fact that $\Phi$ preserves uniformity and $H_{d} - H_{d-1} = \frac{1}{d} = \lVert x \rVert_\infty$.

To study the axiomatic properties and computational aspects of harmonic entropy, we make use of the equivalent infinite-series representation 
\begin{equation} \label{equation:harmonic-entropy-infinite-series}
	F(x) = \sum_{\ell = 0}^{\infty} \left( \frac{1}{\ell + 1} - \lVert \Phi^\ell(x) \rVert_\infty \right).
\end{equation}
Further, let
$
F_L(x)\coloneqq \sum_{\ell = 0}^{L-1} \left( \frac{1}{\ell + 1} - \lVert \Phi^\ell(x) \rVert_\infty \right)
$
denote its \emph{truncation} to the first $L$ terms.
In \Cref{appendix:harmonic-entropy}, we derive this representation by a telescoping argument and establish the properties of $F$ stated below.

\begin{restatable}{proposition}{harmonicentropyproperties}\label{prop:harmonic-entropy-properties}
	$F$ is continuous, concave, symmetric (invariant under coordinate permutations), and invariant under adding zero coordinates.
\end{restatable} 

An important consequence is that harmonic entropy is maximized by uniform distributions. By concavity and symmetry, for every distribution $x \in \Delta^d$, $F(x) \leq F \left(\frac{1}{d}, \dots, \frac{1}{d}\right) = H_{d-1}$.
Thus, harmonic entropy recovers the \pav score under uniform payments, while non-uniform distributions attain weakly lower scores (see also \Cref{appendix:further-harmonic-entropy-results}). Our \emph{harmonic entropy rule} extends this scoring principle by optimizing jointly over voter payment distributions subject to the quota constraints \ref{cond:p4}, thereby accounting for restrictions that \pav does not incorporate.

\paragraph{Connection to Shannon entropy.}
An entropy measure quantifies the uncertainty or diversity of a probability distribution. Entropy-based objectives have appeared in social choice in various settings aiming for proportional outcomes under constraints, see, e.g., \citet{moulin2016entropy, KrPe25b}.
In addition to satisfying similar properties (namely, those in \Cref{prop:harmonic-entropy-properties}), our newly introduced entropy measure has a natural interpretation as a discrete analog of Shannon entropy. 
To see this, extend the water-level representation to $s \in \R_{>0}$ by letting $t_s(x)$ be the unique solution to $s \, t_s(x) = \sum_{j}(x_j - t_s(x))_+$, and set $t_0(x) = \lVert x \rVert_\infty$. Here, $s$ represents continuously varying added capacity. With this, Shannon entropy satisfies
\[
	-\sum_j x_j \ln x_j = \int_0^{\infty} \frac{1}{s+1} - t_s(x) \,\mathrm{d}s, 
\]
where $0 \ln0=0$; see \Cref{appendix:harmonic-entropy} for a proof.
Harmonic entropy sums the same expression over integer steps instead of integrating over continuous capacity.
On uniform distributions, this replaces $\ln d$ by $H_{d-1}$. Thus, harmonic entropy retains the preference for balanced distributions while providing the exact reciprocal marginal gains underlying \pav.

\section{The Harmonic Entropy Rule} \label{section:harmonic-entropy-rule}

We now define our voting rule which jointly optimizes over all full committees and payment systems. 
For a committee $W \subseteq C$, let its \emph{harmonic entropy value} be 
\[
	\mathcal{V}(W) = \max_{(r,p) \in \mathcal P_{124}(W)} \sum_{i \in N} F(r_i,p_i).
\]
Thus, $\mathcal{V}(W)$ is the highest total harmonic entropy achievable by payments satisfying \ref{cond:p1}, \ref{cond:p2}, and \ref{cond:p4}. We call such payment systems \emph{optimal}.
Note that the maximum is always attained as $\mathcal{P}_{124}(W)$ is non-empty and compact,\footnote{The payment system with zero payments and $r_i=1$ for all $i \in N$ is feasible. Additionally, condition \ref{cond:p1} defines a finite product of simplices, while \ref{cond:p2} and \ref{cond:p4} encode closed linear inequalities. Thus, $\mathcal P_{124}(W)$ is compact.}
 and $F$ is continuous.

The global optimization variant of our rule selects committees with the largest such value.

\begin{definition}
	Given an instance $I$, the \emph{harmonic entropy rule} selects all committees $W \in \mathcal W_k$ maximizing $\mathcal{V}(W)$, i.e., $\argmax_{W \in \mathcal W_k} \mathcal{V}(W)$.
\end{definition}

We show that every selected committee admits an optimal payment system additionally satisfying \ref{cond:p3}, establishing core+ existence. Our argument follows the single-candidate-swap argument used for \pav in \Cref{thm:pav-p123}. Hence, global optimality is not needed and it suffices that every swap increases $\mathcal{V}$ by strictly less than $q - \frac{n}{k+1}$. 
For the Hare quota $q = \frac{n}{k}$, this threshold is $\frac{n}{k(k+1)}$.

\paragraph{Adding a candidate.} 
We start by showing that a violation of \ref{cond:p3} identifies a candidate that, when added, increases the harmonic entropy value by at least the quota $q$. 

\begin{theorem} \label{thm:addition}
	Let $W \subseteq C$ be a committee and $(r,p) \in \mathcal P_{124}(W)$ be one of its optimal payment systems. If there exists a candidate $c \in C \setminus W$ with $\sum_{i \in N_c} r_i \ge q$, then $\mathcal V(W \cup \{c\}) \ge \mathcal V(W) + q$.
\end{theorem}

\begin{proof}
	Starting from the system $(r,p) \in \mathcal P_{124}(W)$, we construct a payment system $(r',p')$ for $W' \coloneqq W \cup \{c\}$ by applying the operator $\Phi$ to each voter $i \in N_c$, where we interpret the new coordinate of $\Phi(r_i, p_i)$ as $i$'s payment for $c$. For voters $i \notin N_c$, their reserve and payments are unchanged.
	Then $(r',p')$ satisfies \ref{cond:p1} because $\Phi$ produces probability distributions, and \ref{cond:p2} because each reserve remains a largest coordinate after applying $\Phi$. Since $\Phi$ weakly decreases all coordinates, \ref{cond:p4} is satisfied for all candidates $d \in W$.
	Note further that $F(\Phi(r_i, p_i)) = F(r_i, p_i) + r_i$ by definition. Hence, 
	\[
		\sum_{i \in N} F(r'_i, p'_i) = \sum_{i \in N} F(r_i, p_i) + \sum_{i \in N_c} r_i = \mathcal V(W) + \sum_{i \in N_c} r_i.
	\]
	
	As it may occur that $(r',p')$ violates \ref{cond:p4} for the new candidate $c$, we construct an intermediate payment system.
	To this end, let $R = \sum_{i \in N_c}r_i \ge q$ and set $\lambda = \frac{q}{R} \in (0,1]$. To construct a convex combination of the two payment systems, we need to extend the original payment system $(r, p)$ to $W'$ by assigning zero payments to $c$, i.e., $p_{ic}=0$ for all $i \in N_c$. Due to invariance under adding zero coordinates, this does not affect the harmonic entropy $F$. 
	Now define $(r^\lambda, p^\lambda) \coloneqq \lambda(r', p') + (1-\lambda)(r, p, 0)$.
	Both systems satisfy \ref{cond:p1}, \ref{cond:p2}, and \ref{cond:p4} for every $d \in W$, so their convex combination does as well.
	For the new candidate $c$, $\sum_{i\in N_c} p^\lambda_{ic} = \lambda \sum_{i\in N_c} p'_{ic} \leq \lambda R = q$, so \ref{cond:p4} is also satisfied. Thus, $(r^\lambda, p^\lambda) \in \mathcal{P}_{124}(W')$.
	Using concavity of $F$, we obtain
	\[
		\sum_{i\in N} F(r^\lambda_i, p^\lambda_i) \ge \lambda \sum_{i \in N} F(r'_i,p'_i)
       +(1-\lambda) \sum_{i \in N} F(r_i,p_i) = \sum_{i\in N} F(r_i,p_i) + \lambda \sum_{i\in N_c} r_i = \mathcal V(W) + q.
	\]
	Consequently, $\mathcal{V}(W') \geq \mathcal{V}(W) + q$.
\end{proof}

\paragraph{Deleting a candidate.}
Having added a candidate $c$ to the committee, which is now of size $k+1$, we now identify a candidate $d$ for deletion by constructing a payment system that lower bounds the resulting committee's harmonic entropy value decrease by $\frac{n}{k+1}$. 

\begin{theorem} \label{thm:removal}
	Let $W \subseteq C$ be a committee of size $\lvert W \rvert \ge k+1$. Then there exists a candidate $d \in W$ such that $\mathcal{V}(W \setminus \{d\}) \ge \mathcal{V}(W) - \frac{n}{\lvert W \rvert}$.
\end{theorem}

Throughout the argument, we let $W$ be a committee of size $\lvert W \rvert \ge k+1$. Given a payment system, a candidate $c \in W$ is \emph{unsaturated} if $\sum_{i\in N_c} p_{ic} < q$. A voter is \emph{unsaturated} if she approves at least one such unsaturated candidate in $W$.

To identify a candidate whose deletion causes only a small loss, we first redistribute voters' reserves among their approved committee members, subject to the cap $q$. We continue until no voter with positive reserve approves an unsaturated candidate.
Formally, given $(r,p)\in\mathcal P_{124}(W)$, a \emph{de-reserved} system is an auxiliary payment system $(\hat r, \hat p) \in \mathcal{P}_{14}(W)$ satisfying
\begin{itemize}
	\item $\hat p_{ic} \ge p_{ic}$ for every $i \in N$ and $c \in A_i \cap W$, and 
	\item if a candidate $d \in W$ is \emph{unsaturated} in $(\hat r, \hat p)$, then $\hat r_i=0$ for every $i\in N_d$.
\end{itemize}
Thus, payments can only increase and any remaining reserve must belong to saturated voters. 

\begin{lemma}\label{lem:de-reserved-system}
	For every committee $W \subseteq C$, there exist an optimal payment system $(r,p) \in \mathcal P_{124}(W)$ and a de-reserved system $(\hat r,\hat p)$ derived from it such that, for every unsaturated candidate $d \in W$ and every voter $i \in N_d$,
	\begin{enumerate}
		\item[(a)] $\hat p_{id} \ge \hat p_{ic}$ for every $c \in A_i \cap W$, and
		\item[(b)] $F(\hat r_i, \hat p_i) = F(r_i, p_i) - \hat p_{id}$.
	\end{enumerate}
\end{lemma}

We prove the lemma at the end of the section and focus on proving \Cref{thm:removal} from it. 
Let $(r, p)$ be an optimal payment system for the committee $W$ and $(\hat r, \hat p)$ its de-reserved system chosen as in \Cref{lem:de-reserved-system}.

Let $d \in W$ be a candidate with smallest total payment $\sum_{i \in N_d} \hat p_{id}$ in the de-reserved system. By \ref{cond:p1}, the total payments in $(\hat r, \hat p)$ are at most $n$, and thus, the total payments to $d$ satisfy $\sum_{i \in N_d} \hat p_{id} \le \frac{n}{\lvert W \rvert} < q$.
Hence, $d$ is unsaturated.
Deleting $d$ gives rise to $W' = W \setminus \{d\}$. We will now build a payment system for $W'$ serving as a witness for the inequality $\mathcal{V}(W') \ge \mathcal{V}(W) - \frac{n}{\lvert W \rvert}$, establishing our claim in \Cref{thm:removal}. To this end, let the payment system $(r', p') \in \mathcal P_{124}(W')$ be derived from the de-reserved system $(\hat r, \hat p)$ by
\[
r'_i = \begin{cases}
	\hat p_{id} & \text{if $i \in N_d$,} \\
	r_i & \text{if $i \not \in N_d$,}
\end{cases}
\qquad
\text{and for all $c \in A_i \cap W'$, }
\quad
p'_{ic} =
\begin{cases}
	\hat p_{ic} & \text{if $i \in N_d$,} \\
	p_{ic} & \text{if $i \not \in N_d$.}
\end{cases}
\]
In other words, for every voter $i \in N_d$, her payment distribution $(r'_i, p'_i)$ is obtained from $(\hat r_i, \hat p_i)$ by deleting the zero reserve coordinate and relabeling the payment to $d$ as the new reserve. Hence, this preserves unit mass and by \Cref{lem:de-reserved-system}(a), \ref{cond:p1} and \ref{cond:p2} are satisfied.
For every voter $i \notin N_d$, we leave the original payment distribution $(r_i, p_i)$ unchanged.
Hence, \ref{cond:p1} and \ref{cond:p2} are still satisfied for these voters.
Lastly, \ref{cond:p4} is also satisfied because the de-reserved system satisfies \ref{cond:p4} and for $i \notin N_d$, we have $p'_{ic} = p_{ic} \le \hat p_{ic}$ for every candidate $c \in A_i \cap W \setminus \{d\}$.
Consequently, $(r', p') \in \mathcal P_{124}(W')$.

Finally, we can bound $\mathcal V(W')$ by calculating the sum of voters' harmonic entropies for the payment system $(r', p')$. For every voter $i \not \in N_d$, we have $F(r'_i, p'_i) = F(r_i, p_i)$ because $(r'_i, p'_i) = (r_i, p_i)$, and for $i \in N_d$, we have $F(r'_i, p'_i) = F(\hat r_i, \hat p_i)$ because $(r'_i, p'_i)$ is obtained from $(\hat r_i, \hat p_i)$ by dropping a $0$-coordinate and permuting the remaining coordinates, which does not change harmonic entropy by \Cref{prop:harmonic-entropy-properties}. In addition, by \Cref{lem:de-reserved-system}(b), $F(\hat r_i, \hat p_i) = F(r_i, p_i) - \hat p_{id}$. Hence,
\[
	\mathcal{V}(W') = \max_{(\bar r, \bar p) \in \mathcal P_{124}(W')} \sum_{i \in N} F(\bar r_i,\bar p_i) \ge \sum_{i \in N} F(r'_i, p'_i) = \sum_{i \in N} F(r_i, p_i) - \sum_{i \in N_d} \hat p_{id} \ge \mathcal{V}(W) - \frac{n}{\lvert W \rvert},
\]
recalling that $d$ was chosen such that $\sum_{i \in N_d} \hat p_{id} \le \frac{n}{\lvert W \rvert}$. This proves \Cref{thm:removal} and it remains to prove \Cref{lem:de-reserved-system}.

\begin{proof}[Proof of \Cref{lem:de-reserved-system}]
	Recall that $\mathcal{P}_{124}(W)$ is non-empty, compact and convex since it is described by linear inequalities. Thus, by continuity and concavity of $F$ (\Cref{prop:harmonic-entropy-properties}), the maximizers of $\sum_{i \in N} F(r_i,p_i)$ over $\mathcal{P}_{124}(W)$ form a compact and convex set. Let $(r,p)$ be a maximizer that minimizes the total reserve $\sum_{i \in N} r_i$; this exists by the properties we've discussed.  
	For this particular choice of optimal payment system, we claim that
	\begin{align}\label{eq:claimde-reserved}
		\text{if a candidate $d \in W$ is unsaturated under $p$, then $p_{id} = r_i$ for every $i \in N_d$.}
	\end{align}
	Assume for a contradiction that for some unsaturated candidate $d \in W$ and some $i \in N_d$, we have $p_{id} < r_i$. For a sufficiently small $\varepsilon > 0$, decrease $r_i$ by $\varepsilon$, decrease each $p_{ic}$ satisfying $p_{ic} = r_i$ by $\varepsilon$, and increase $p_{id}$ by the sum of all these decreases to preserve \ref{cond:p1}. These changes keep the reserve $r_i - \varepsilon$ as a largest coordinate (preserving \ref{cond:p2}), and they weakly decrease payments to all $c \neq d$ so that \ref{cond:p4} still holds. Choosing $\varepsilon$ small enough ensures that payments remain non-negative. Finally, such a modification never decreases $F$ since it is concave.
	To see this, moving $\varepsilon$ from a larger coordinate to $p_{id}$ is a convex combination of the original payments and a payment distribution where we swapped that larger coordinate and $p_{id}$. Note that the latter vector yields the same value for $F$ by symmetry (see \Cref{prop:harmonic-entropy-properties}). It follows that after these changes, the payments are still optimal (maximizing $\sum_{i \in N} F(r_i,p_i)$). However, they strictly decreased $\sum_{j \in N} r_j$, contradicting our choice of $(r,p)$.
	
	We now construct a de-reserved system $(\hat r, \hat p) \in \mathcal{P}_{14}(W)$ from $(r,p)$ using a continuous process that we describe next. Throughout the process, let $U = \{ c \in W : \sum_{i \in N_c} \hat p_{ic} < q \}$ denote the set of currently unsaturated candidates. We call a voter $i \in N$ \emph{active} if $\hat r_i > 0$ and $A_i \cap U \neq \emptyset$. Note that an active voter $i$'s reserve can be reduced by using it to make additional payments to unsaturated candidates.
	
	Start the process at $(\hat r, \hat p)=(r, p)$. While there is at least one active voter, let each active voter $i$ continuously increase $\hat p_{ic}$ at unit rate for all $c \in A_i \cap U$, while decreasing $\hat r_i$ at rate $\lvert A_i \cap U \rvert$ to preserve \ref{cond:p1}; this continues until either some active voter's reserve $\hat r_i$ reaches $0$, or some candidate $c \in U$ becomes saturated. Such an event concludes a stage of the process, at which point we update $U$ by removing saturated candidates and update voters from active to non-active. We then continue the process with the next stage should active voters remain.
	Since each stage either exhausts a voter's reserve (turning them non-active) or saturates a candidate, the process terminates after at most $n + \lvert W \rvert$ stages. At termination, we have a payment system $(\hat r, \hat p) \in \mathcal{P}_{14}(W)$ satisfying $\hat p_{ic} \ge p_{ic}$ for all $i \in N$ and $c \in A_i\cap W$, and no active voters by construction. Hence, $(\hat r, \hat p)$ forms a de-reserved system.
	
	It remains to check properties (a) and (b) in the lemma statement. At initialization, (a) holds because $\hat p_{id} = r_i \ge \hat p_{ic}$ for all $c \in A_i \cap W$ by our claim \eqref{eq:claimde-reserved} and \ref{cond:p2}. Since only payments to candidates currently in $A_i \cap U$ ever increase, and they all increase at the same rate, $\hat p_{id} \ge \hat p_{ic}$ continues to hold throughout the process for all $c \in A_i \cap W$, so (a) continues to hold throughout.
	For (b), note that all payment increases occur above level $r_i$: Any $c \in A_i \cap W$ with $p_{ic} < r_i$ is not in $U$, and thus, $\hat p_{ic} = p_{ic} < r_i$. Hence, capping $\hat p_i$ at level $r_i$ recovers $p_i$, i.e., $\min(\hat p_{ic}, r_i) = p_{ic}$ for all $c \in A_i \cap W$, and the increase above level $r_i$ is $(\hat p_{ic} - r_i)_+ = \hat p_{ic} - p_{ic}$. Since $\hat r_i = 0$, voter $i$'s entire initial reserve $r_i$ was added to $\hat p_i$, so we have $\sum_{c \in A_i \cap W} (\hat p_{ic} - r_i)_+ = \sum_{c \in A_i \cap W} (\hat p_{ic} - p_{ic}) = r_i$. By definition of $\Phi$ (with level $t = r_i$), this means $\Phi(\hat p_i) = (p_i, r_i)$, which corresponds to $(r_i, p_i)$ with reordered coordinates. Since $\lVert \hat p_i \rVert_\infty = \hat p_{id}$ by (a), the Robin Hood equality and symmetry of $F$ give $F(r_i, p_i) = F(\hat p_i) + \hat p_{id} = F(\hat r_i, \hat p_i) + \hat p_{id}$, where we used invariance under adding zero coordinates.
\end{proof}

\paragraph{Core+ existence.}
We can now combine the addition and deletion arguments to show that the harmonic entropy rule selects committees satisfying core+. In fact, we will show that core+ is even satisfied by committees whose harmonic entropy objective cannot be increased much by a single-candidate swap, i.e., by approximate local optima of the harmonic entropy objective. 

\begin{theorem}
	\label{thm:local-opt-satisfy-core+}
	Let $\frac{n}{k+1}<q\le\frac{n}{k}$ and $\delta = q - \frac{n}{k+1} > 0$. Let $W \in \mathcal W_k$ such that for all $d \in W$ and $c \in C \setminus W$, we have $\mathcal V(W) > \mathcal V(W \cup \{c\} \setminus \{d\}) - \delta$. Then $W$ satisfies core+. In particular, a committee $W$ selected by the harmonic entropy rule satisfies core+.
\end{theorem}
\begin{proof}
	We prove the contrapositive. Let $W$ be a committee with $\lvert W \rvert = k$ and let $(r,p) \in \mathcal{P}_{124}(W)$ be an optimal payment system for it. Assume that $W$ fails core+, so $(r,p)$ fails \ref{cond:p3}; hence there exists a losing candidate $c \in C \setminus W$ with $\sum_{i \in N_c} r_i \ge q$. By \Cref{thm:addition}, $\mathcal V(W \cup \{c\}) \ge \mathcal V(W) + q$. Since $\lvert W \cup \{c\} \rvert = k+1$, we can apply \Cref{thm:removal} and deduce that there exists a candidate $d \in W \cup \{c\}$ such that $\mathcal{V}(W \cup \{c\} \setminus \{d\}) \ge \mathcal{V}(W \cup \{c\}) - \frac{n}{k+1}$. Hence,
	\[
	\mathcal{V}(W \cup \{c\} \setminus \{d\}) \ge \mathcal{V}(W \cup \{c\}) - \tfrac{n}{k+1} \ge \mathcal V(W) + q - \tfrac{n}{k+1}  = \mathcal V(W) + \delta.
	\]
	Note that we must have $c \neq d$ since the swap increases the score by $\delta > 0$. This proves the contrapositive.
\end{proof}

\paragraph{Other properties of the harmonic entropy rule.}
In \Cref{appendix:further-harmonic-entropy-results}, we discuss some additional interesting properties of the harmonic entropy rule. We show that in addition to core+, it also guarantees Pareto optimality, i.e., there does not exist committee $W' \in \mathcal W_k$ such that $u_i(W') \ge u_i(W)$ for all $i \in N$, with strict inequality for at least one $i \in N$.

\begin{restatable}{proposition}{heroParetoOptimal} \label{prop:hero-pareto-optimal}
	Every committee selected by the harmonic entropy rule is Pareto optimal.
\end{restatable}

However, the harmonic entropy rule must fail fractional Pareto optimality (fPO, \citealp{BBF+26fractional,PVY26representative}) since this property is incompatible with Justified Representation \citep[Theorem 1]{jin2026limitationsbobw}. In addition, the harmonic entropy rule fails priceability as defined by \citet{PeSk20a}. We give an example in \Cref{appendix:further-harmonic-entropy-results}. Thus, it remains an open problem whether Pareto optimality and priceability are compatible.

In the appendix, we also discuss the relationship between the harmonic entropy rule and \pav. The latter can equivalently be defined as maximizing harmonic entropy over $\mathcal{P}_{123}$ rather than over $\mathcal{P}_{124}$.

\begin{proposition}
	\pav can be equivalently defined as the rule that selects exactly those committees $W$ maximizing $\max_{(r,p) \in \mathcal P_{123}(W)} \sum_{i \in N} F(r_i,p_i)$.
\end{proposition}

In addition, the harmonic entropy rule coincides with \pav for the special case of \emph{party-approval elections} \citep{BGP+24a,golz2026apportionment} where every candidate has at least $k$ copies.
\begin{restatable}{proposition}{heroPartyApproval}\label{prop:hero-party-approval}
	In party-approval instances, \pav and the harmonic entropy rule select the same committees.
\end{restatable}
While we show in the next section that a core+ outcome can be computed in polynomial time (by local search over the harmonic entropy objective), computing a globally optimal committee for the harmonic entropy rule is NP-hard. 
This follows because the winner determination problem for \pav is NP-complete even in the party-approval case \citep[Theorem 5.1]{BGP+24a}.
\begin{restatable}{proposition}{heroNPComplete}\label{prop:hero-np-complete}
	Given an instance and a number $B \ge 0$, the problem of deciding whether there exists a committee $W \in \mathcal W_k$ with $\mathcal{V}(W) \ge B$ is NP-complete, even for party-approval instances.
\end{restatable}
NP-hardness (under Turing reductions) of computing a winning committee of the harmonic entropy rule could also be deduced from \citet{dong2026jointly}, who shows that, unless P = NP, no polynomial-time algorithm can satisfy both Pareto optimality and Justified Representation (a property much weaker than core+).

\section{Polynomial-Time Computation} \label{section:computation}

\begin{wrapstuff}[
    r,
    top=3,
    width=0.48\textwidth,
    hsep=1.2em,
    vsep=0.4em
]
\begin{minipage}{\linewidth}
\begin{algorithm}[H]
    \caption{\mbox{LS-Harmonic Entropy Rule}}
    \label{alg:local-search}

    \KwIn{$(A_i)_{i\in N}$, $k$, $q \in \mathbb{Q} \cap (\frac{n}{k+1},\frac{n}{k}]$}
    \KwOut{committee $W \in \mathcal{W}_k$}

    $\delta \gets q-\frac{n}{k+1}$\;
	$L \gets \lceil 2nk/\delta \rceil$\;
    Initialize $W\in\mathcal{W}_k$ arbitrarily\;

    \While{$\exists c\in W,\ d\in C \setminus W:
        \mathcal V_L((W \setminus \{c\}) \cup \{d\})-\mathcal V_L(W)\ge \delta/2$}{
        $W\gets (W \setminus \{c\}) \cup \{d\}$\;
    }
    \Return{$W$}\;
\end{algorithm}
\end{minipage}
\end{wrapstuff}

In this section, we show that committees satisfying core+ can be computed in polynomial time in the input size and $\frac{1}{\delta}$, where $\delta=q-\frac{n}{k+1}>0$.
The local-optimum part of \Cref{thm:local-opt-satisfy-core+} shows that committees that only admit single-candidate swaps increasing the harmonic entropy objective by less than $\delta$ must satisfy core+.\footnote{The gap $\delta$ is important for the local search to terminate in polynomially many steps. \citet{KrEl24a} give an example for \pav where it takes superpolynomially many swaps to reach a local optimum with $\delta = 0$, and this example can be adapted to work for harmonic entropy, too.}

With this in mind, we can adapt the local-search \pav algorithm proposed by \citet{Aziz2018complexity}: start from an arbitrary committee $W \in \mathcal{W}_k$ and perform single-candidate swaps that increase $\mathcal V(W)$ by at least $\delta/2$. 
Since harmonic entropy is bounded by the $k$-th harmonic number, the objective function $\mathcal V$ is bounded by $nH_k$. Hence, the algorithm can conduct at most $\lfloor 2nH_k/\delta\rfloor$ swaps until this bound is attained.
The complete algorithm is shown in \Cref{alg:local-search}.
To implement this algorithm in polynomial time, we need to evaluate (or at least closely approximate) $\mathcal V(W)$ in polynomial time.
We start by deriving a linear-programming approximation for our harmonic entropy objective and only afterward argue that sufficiently close approximations detect whether the single-swap condition, ensuring core+, is satisfied.

A natural way to approximate the harmonic entropy is to use its \emph{truncated} version 
\[
	F_L(r_i,p_i)\coloneqq \sum_{\ell = 0}^{L-1} \left( \frac{1}{\ell + 1} - \lVert \Phi^\ell(r_i,p_i) \rVert_\infty \right)
	=\sum_{\ell = 0}^{L-1} \left( \frac{1}{\ell + 1} - t_\ell(r_i,p_i) \right).
\] 
The corresponding \emph{truncated committee value} is then defined by
\[
	\mathcal V_L(W) \coloneqq \max_{(r, p) \in \mathcal P_{124}(W)} \sum_{i \in N}F_L(r_i, p_i).
\]

The truncation allows us to write this optimization problem as a linear program with polynomially many variables and constraints, provided that $L$ is polynomially bounded. To represent $t_\ell(r_i,p_i)$, we use the identity \eqref{eq:water-level-minimum} from \Cref{lem:water-level-characterization} and the fact that the constraint $t\geq t_\ell(x)$ is equivalent to $\sum_j(x_j-t)_+\leq\ell \cdot t$.
Introducing one new variable $a_j$ for each coordinate, it becomes
\[
    a_j\geq x_j-t,\qquad a_j\geq0,\qquad \sum_j a_j\leq\ell \cdot t.
\]
The first two inequalities encode $a_j \ge (x_j-t)_+$. Conversely, whenever the original constraint holds,
choosing $a_j=(x_j-t)_+$ satisfies all constraints. Since the optimum increases as each water-level variable decreases, we obtain $t_{i\ell} = \lVert\Phi^\ell(r_i, p_i)\rVert_\infty$ at optimality.
Thus, we obtain the following linear program for $\mathcal V_L(W)$ with $O(nkL)$ variables and constraints:
\begin{equation} \label{eq:harmonic-entropy-LP}
	\begin{aligned}
		\max \quad 			& \sum_{i \in N} \sum_{\ell = 0}^{L - 1} \left( \frac{1}{\ell + 1} - t_{i\ell} \right) \\
		\text{s.t.} \quad 	& (r, p) \in \mathcal{P}_{124}(W), \\
							& a_{i\ell r} \ge r_i - t_{i\ell}, \quad a_{i\ell c} \ge p_{ic} - t_{i\ell} && \forall i \in N,\; 0 \le \ell < L,\; c \in A_i \cap W, \\
							& a_{i\ell r} + \sum_{c \in A_i \cap W} a_{i\ell c} \le \ell \cdot t_{i\ell} && \forall i \in N,\; 0 \le \ell < L, \\
							& a_{i\ell r},\, a_{i\ell c},\, t_{i\ell} \ge 0.
	\end{aligned}
\end{equation}

As a next step, we establish a uniform bound on the approximation error incurred by truncating the infinite series to $L$ summands. Since this bound decreases as $L$ increases, we can choose $L$ sufficiently large to ensure that improvements in the objective of at least $\delta$ do not go unnoticed. This will complete our argument.

\begin{lemma}\label{lem:entropy-approximation-bound}
	For every committee $W \in \mathcal W_k$, we have $0 \leq \mathcal V(W) - \mathcal V_L(W) \leq \frac{nk}{L}$.
\end{lemma}

\begin{proof}
	By \Cref{lem:harmonic-entropy-tail}, truncating the infinite series $F(r_i,p_i)$ after $L$ terms incurs an error of at most $\frac{d-1}{L}$ for any distribution $(r_i, p_i) \in \Delta^d$.
	Hence, we can just aggregate these errors over all voters. For a committee $W \in \mathcal W_k$, the distribution $(r_i, p_i)$ associated with each voter $i \in N$ has at most $d_i = \lvert A_i \cap W \rvert + 1 \leq k+1$ coordinates, so $0 \leq F(r_i, p_i) - F_L(r_i, p_i) \leq \frac{k}{L}$. Therefore, the approximation error of our objective function is bounded by 
	\[
		0 \leq \mathcal V(W) - \mathcal V_L(W)
		\leq \max_{(r,p) \in \mathcal P_{124}(W)} \sum_{i \in N} \left( F(r_i, p_i) - F_L(r_i, p_i) \right)
		\leq \frac{nk}{L}. \qedhere
	\]
\end{proof}

With the help of \Cref{lem:entropy-approximation-bound}, we can now show that \Cref{alg:local-search} computes a core+ committee in polynomial time.

\begin{theorem} \label{theorem:core-polynomial}
	Set $L = \lceil 2nk/\delta \rceil$. \Cref{alg:local-search} computes a core+ committee and runs in polynomial time in the input size and $\frac{1}{\delta}$.
\end{theorem}

\begin{proof}
	By definition, it follows that $0 \le \mathcal V_L(W) \leq nH_k$ for every $W \in \mathcal W_k$. Since each single-candidate swap increases the objective $\mathcal V_L(W)$ by at least $\delta/2$, there can be at most $\frac{2nH_k}{\delta}$ swaps.

	After termination of \Cref{alg:local-search}, we obtain $\mathcal V_L(W') - \mathcal V_L(W) < \delta/2$ for all $W' = (W \cup \{d\}) \setminus \{c\}$ with $d \in C\setminus W$ and $c \in W$. Our choice of $L$ yields
	\[
		\mathcal V(W') - \mathcal V(W) \leq \left(\mathcal{V}_L(W')+\frac{nk}{L}\right) - \mathcal V_L(W)
		< \frac{\delta}{2} + \frac{\delta}{2} = \delta,
	\]
	where we upper-bounded $\mathcal V(W')$ and lower-bounded $\mathcal V(W)$ with the help of \Cref{lem:entropy-approximation-bound}.
	Hence, the final committee $W$ satisfies core+ by \Cref{thm:local-opt-satisfy-core+}. Each iteration requires solving at most $k(m-k)+1$ rational linear programs, where each of them has $O(nkL)$ variables and constraints. Together with the bound on the number of swaps, this establishes the claimed polynomial dependence on the input size and $\frac{1}{\delta}$.
\end{proof}

For the Hare quota, $\delta = \frac{n}{k(k+1)}$, so the algorithm runs in polynomial time. Moreover, we prove in \Cref{appendix:computing-entropy-objective} that choosing $L = (k+1)^4$ suffices so that the linear program in \eqref{eq:harmonic-entropy-LP}, together with an explicit correction term, \emph{exactly} computes $\mathcal V(W)$ for each committee $W \in \mathcal{W}_k$.

\section{Extension to Droop Core+} \label{section:droop}

So far, we have restricted ourselves to quotas strictly above the Droop quota $q_D = \frac{n}{k+1}$. We now extend our results to this specific threshold, where affordability is defined with a strict inequality. This distinction is essential: allowing deviations at equality can make the core empty.

A committee $W \in \mathcal W$ is in the \emph{Droop core} if no non-empty coalition $S \subseteq N$ with proposal $T \subseteq C$ satisfying $\lvert S \rvert > q_D \lvert T \rvert$ strictly increases every member's utility, i.e., $u_i(T) \ge u_i(W)+1$ for all $i \in S$.
Similarly, $W$ satisfies \emph{Droop core+} if there is no pair consisting of participation vector $y \in \R_{\geq 0}^n \setminus \{\mathbf 0\}$ and candidate quantities $z\in\R_{\geq0}^m$ that satisfies the improvement condition from \eqref{eq:fraccoredeviation} and 
\[
\sum_{i \in N} y_i > q_D \sum_{c \in C} z_c.
\] 
As seen in \Cref{sec:core+} for quotas in $(\frac{n}{k+1},\frac{n}{k}]$, restricting $y$ and $z$ to binary vectors recovers ordinary blocking coalitions. Thus, Droop core+ implies the Droop core.

Every Droop deviation remains affordable at quotas sufficiently close to $q_D$ from above. Since there are only finitely many committees, our previous existence result for these quotas extends to Droop core+.

\begin{theorem}\label{thm:droop}
	Every election instance $I$ admits a Droop core+ committee.
\end{theorem}

\begin{proof}
	Choose a sequence of quotas $q_r \in (\frac{n}{k+1}, \frac{n}{k}]$, converging to the Droop quota from above. By \Cref{thm:local-opt-satisfy-core+}, for each $r$, there exists a committee $W_r \in \mathcal W_k$ satisfying core+ at quota $q_r$. 
	Since there are finitely many committees, there is a committee $W^*$ which occurs infinitely often.

	Suppose that $y \in \R_{\geq 0}^n \setminus \{\mathbf 0\}$ and $z \in \R_{\ge0}^m$ form a fractional blocking coalition against $W^*$ with respect to our Droop core+ definition. Hence, strict affordability gives $\sum_i y_i > q_D \sum_c z_c$.
	For sufficiently large $r$, we therefore have $\sum_i y_i \geq q_r \sum_c z_c$.
	Selecting a sufficiently large index $r$ for which $W_r = W^*$ gives a fractional coalition blocking at quota $q_r$ with respect to core+; a contradiction.
\end{proof}

While the existence of core+ extends to the stronger Droop core+, the computational results do not extend to this endpoint, the reason being that our local-search algorithm requires a positive gap $\delta = q - q_D > 0$. That said, integrality of the ordinary Droop core still allows polynomial computability. 
The key observation, recorded by \citet[Lemma~1]{casey2025droop}, is that for positive integers $s$ and $r$, 
\[
	s > \frac{r n}{k+1} \quad\Longrightarrow\quad s \geq \frac{r n+1}{k+1}.
\]
Thus, strict affordability has a positive minimum excess when coalition and proposal sizes are integral. This gives the following theorem.

\begin{theorem}
	A committee in the ordinary Droop core can be computed in polynomial time.
\end{theorem}

\begin{proof}
	Run \Cref{alg:local-search} with quota $q = q_D + \frac{1}{k(k+1)}$. This quota lies in $(\frac{n}{k+1}, \frac{n}{k}]$.
	By \Cref{theorem:core-polynomial}, the algorithm returns a core+ committee $W$ with respect to the quota $q$ in polynomial time since $\frac{1}{q-q_D} = k(k+1)$.

	Now suppose that coalition $S$ blocks $W$ with proposal $T$ with respect to the ordinary Droop core. Note that $1 \le \lvert T \rvert\le k$ by non-emptiness and $\frac{n\lvert T \rvert}{k+1}<\lvert S \rvert\le n$. Together with the above observation by \citet{casey2025droop}, this yields
	\[
		\frac{\lvert S \rvert}{\lvert T \rvert} \ge q_D + \frac{1}{\lvert T \rvert (k+1)} \ge q_D + \frac{1}{k(k+1)} = q.
	\]
	Thus, the same coalition can afford $T$ at quota $q$, contradicting core+.
\end{proof}

For integral deviations, a small increase in quota preserves every Droop blocking coalition. A fractional deviation can, however, exceed the Droop threshold by an arbitrarily small amount, so the same argument does not establish the result for the exact Droop core+. Nevertheless, we can choose the quota closer to $q_D$ which yields arbitrarily close approximations. 
More specifically, for every rational $\varepsilon>0$, run \Cref{alg:local-search} with $q = (1+\min\{\varepsilon, \frac{1}{k}\})q_D$ to compute a $(1+\varepsilon)$-approximation to Droop core+ in time polynomial in the input size and $\frac{1}{\varepsilon}$.\footnote{A committee satisfies $(1 + \varepsilon)$-Droop core+ if there is no $(y, z) \in \R_{\geq0}^n\setminus\{\mathbf0\} \times \R_{\geq0}^m$ satisfying
\eqref{eq:fraccoredeviation} and $\sum_i y_i>(1+\varepsilon)q_D\sum_c z_c$.}

\section{Discussion}
We introduce a new voting rule based on the novel concept of \emph{harmonic entropy}. Our rule connects discrete global optimization over committees with continuous optimization over payment allocations. Committees returned by the harmonic entropy rule satisfy core+ and are Pareto optimal. Moreover, while computing a globally optimal committee is NP-hard, a core+ committees can be computed in polynomial time by performing local search on the harmonic entropy objective.

Our techniques give rise to all necessary structural requirements needed to satisfy core+. It would be interesting to see to what extent these conditions already fix our objective to the harmonic entropy rule. In \Cref{tab:payment-certificates}, we provide a short overview of the conditions satisfied by commonly analyzed voting rules in approval-based committee elections (see also \Cref{appendix:discussion}). 

\begin{table}[t!]
    \centering
    \begingroup
    \definecolor{paymentblue}{RGB}{45,105,145}
    \newcommand{\certyes}{\textcolor{paymentblue}{\ensuremath{\checkmark}}}
    \newcommand{\certno}{\textcolor{black!55}{\ensuremath{\times}}}
    \renewcommand{\arraystretch}{1.12}
    \setlength{\tabcolsep}{7pt}
    \small

	\caption{Payment certificates at the Hare quota $q=\frac{n}{k}$.
	A checkmark means that every committee returned by the rule admits a payment system satisfying all conditions in the column.
	Note that every committee trivially admits a certificate satisfying \ref{cond:p1}, \ref{cond:p2}, and \ref{cond:p4}. ($^*$: Filling up the committee with affordable candidates.)}
    \label{tab:payment-certificates}

    \begin{tabular}{@{}lcccc@{}}
        \toprule
        & \multicolumn{4}{c}{Guaranteed payment certificates} \\
        \cmidrule(l){2-5}
        Voting rule
        & \ref{cond:p1}, \ref{cond:p3}
        & \ref{cond:p1}--\ref{cond:p3}
        & \ref{cond:p1}, \ref{cond:p3}, \ref{cond:p4}
        & \ref{cond:p1}--\ref{cond:p4} \\
        \midrule
        Method of Equal Shares (MES)
        & \certyes & \certno & \certyes & \certno \\
        Sequential Phragm\'en
        & \certyes & \certno & \certyes & \certno \\
        Greedy Justified Candidate Rule (GJCR)
        & \certyes & \certno & \certyes\rlap{$^*$} & \certno \\
		Proportional Approval Voting (\pav)
        & \certyes & \certyes & \certno & \certno \\
        Sequential \pav
        & \certno & \certno & \certno & \certno \\
        \midrule
        \textcolor{paymentblue}{Harmonic Entropy Rule}
        & \certyes & \certyes & \certyes & \certyes \\
        \bottomrule
    \end{tabular}
    \endgroup
\end{table}

Harmonic entropy provides a way to evaluate committee payment systems by rewarding balanced budget allocations.
We believe that this perspective may be useful beyond the present core-existence proof, both for designing proportional voting rules and for studying how proportionality interacts with other axiomatic properties.
Its connection to Shannon entropy also suggests exploring similar objectives in related allocation problems.

An interesting open question concerns the exact computation of Droop-core+ committees in polynomial time. More broadly, the existence of the core shifts the attention from whether core+ committees exist to which ones should be selected. This can be, for example, done by introducing additional desiderata. Lastly, understanding the structure of the set of core+ committees may provide additional foundation for creating new voting rules and understanding trade-offs among competing objectives.

\section*{Acknowledgments and AI Usage}

For insightful discussions, we thank Fabian Frank (regarding adaptive payment schemes), Jannik Peters (regarding core+), and Paul Gölz (regarding connections to \pav).

The voting rule we present and the proof that it satisfies core+ were found by GPT-6 Astra in a lengthy interactive session. We initially aimed to find an approximate core outcome. Starting from a Lindahl-equilibrium rounding approach, GPT-6 Astra achieved an approximation factor around $2.065$. Repeated prompts to improve the bound, seek alternative potential functions, and exploit KKT conditions of Lindahl equilibria eventually led to the entropy-based framework in which we optimize continuous voter payments and compare candidate additions with capacity-preserving deletions. We suggested aiming for core+ rather than merely core. Next to that, our work on the exposition was also aided by Claude Opus 5.5 (who suggested, for example, the recursive definition of $F$) and Gemini 4 Argon (who helped us better understand the deletion argument and equivalent definitions of core+). Despite these ideas being contributed by generative AI models, we devoted substantial effort to independently verifying them, making them accessible to human readers (particularly by rewriting the proofs) and investigating possible simplifications, connections to the existing literature, extensions of the results, and directions for further research. 

\printbibliography

\newpage
\addtocontents{toc}{\protect\setcounter{tocdepth}{-1}}
\appendix
\startcontents[appendices]
\section*{Appendix Contents}
\begingroup
\setcounter{tocdepth}{2}
\renewcommand{\setcounter}[2]{}
\printcontents[appendices]{}{1}{}
\endgroup

\section{Core+ and Logical Relationships}\label{app:core+proofs}

\subsection{Equivalent Definitions of Core+}

\corepaymentequiv*

\begin{proof}
	Let $W \in \mathcal W$ be a committee.
	
	For the direction from left to right, assume that $W$ satisfies core+. We will prove existence of a payment system $(r,p) \in \mathcal P_{1234}(W)$ via Farkas' lemma.
	
	For that, we start by transforming the conditions for a beneficial deviation of a fractional blocking coalition from \Cref{def:core+-deviation} into a linear feasibility problem. First, since inequalities are homogeneous under non-negative scaling, we can replace the condition $y \neq 0$ by $\sum_{i \in N}y_i=1$. Second, we introduce an auxiliary variable $x_{ic}$ for every $i \in N$ and $c \in A_i \cap W$ to replace $\min\{y_i,z_c\}$. For each such variable, this gives the additional constraints $x_{ic}\le y_i$ and $x_{ic}\le z_c$. As we want the left-hand side of (\ref{eq:fraccoredeviation}) to be large, we can assume w.l.o.g. that $x_{ic}=\min\{y_i,z_c\}$.
	
	All in all, we get the following (primal) system of (in-)equalities and its corresponding dual:
	{\footnotesize
		\[
		\begin{array}{ccc}
			\hline
			& \text{primal} & \text{dual} \\
			\hline
			\\[-0.8em]
			\forall i \in N, c \in A_i \cap W: &  x_{ic}  \ge  0   & b_{ic}+p_{ic}-r_i  \ge  0 \\[0.5em]
			
			\forall i \in N: &  y_i  \ge  0 &  \hspace{-3em}(u_i(W) + 1) r_i - \sum_{c \in A_i \cap W} b_{ic}- \alpha -\delta  \ge  0 \\[0.5em]
			
			\forall c \in W: &  z_c  \ge  0  & -\sum_{i \in N_c}p_{ic}+q\alpha  \ge  0 \\[0.5em]
			
			\forall c \in C \setminus W: &  z_c  \ge  0  & q \alpha-\sum_{i \in N_c}r_i  \ge  0 \\[0.5em]
			
			&  -\sum_{i \in N}y_i  =  -1 &  \delta  >  0 \\[0.5em]
			\forall i \in N, c \in A_i \cap W:&  x_{ic}-y_i  \le  0 &  b_{ic}  \ge  0 \\[0.5em]
			\forall i \in N, c \in A_i \cap W:&  x_{ic}-z_c  \le  0 &  p_{ic}  \ge  0 \\[0.5em]
			&   q\sum_{c \in C} z_c-\sum_{i \in N} y_i  \le  0 &  \alpha  \ge  0 \\[0.5em]
			\forall i \in N:&   \left(u_i(W)+1\right)y_i-\sum_{c \in A_i \setminus W} z_c -\sum_{c \in A_i \cap W} x_{ic}  \le  0 &  r_i  \ge  0 \\[0.5em]
			\hline
		\end{array}
		\]
	}
	Since $W$ satisfies core+ by assumption, the primal is infeasible. By Farkas' lemma, the dual admits a feasible solution $(\delta,b,p,\alpha,r)$. 
	
	We now modify that solution into a payment system satisfying \ref{cond:p1}--\ref{cond:p4}.
	The dual system is homogeneous, so we can rescale its solution such that $\alpha + \delta = 1$. This immediately implies $\alpha < 1$ by $\delta > 0$ and $r_i > 0$ because $(u_i(W)+1)r_i - \sum_{c \in A_i \cap W} b_{ic}\ge \alpha +\delta$ for all $i \in N$.
	First, without violating any constraints of the dual, we can reduce voters' payments such that $p_{ic} = \min \{p_{ic}, r_i\}$ and set $b_{ic} = r_i - p_{ic}$ for candidates $c \in A_i \cap W$. 
	This can neither increase $p_{ic}$ nor $b_{ic}$: if $p_{ic}$ is not decreased, then $b_{ic} \ge r_i-p_{ic}$ initially holds; if $p_{ic}$ is decreased to $r_i$, $b_{ic}$ is set to $0$. Note that decreasing $p_{ic}$ or $b_{ic}$ can never violate inequalities from the second or third rows of the dual. 
	
	All in all, $p_{ic}\le r_i$ holds for all $i \in N$ and $c \in A_i \cap W$ and thus, \ref{cond:p2} holds.
	Moreover, for every $c \not\in W$, we get $\sum_{i \in N_c}r_i \leq \alpha q < q$ by the fourth row of the dual. This gives \ref{cond:p3}.
	For $c \in W$, the third row of the dual gives \ref{cond:p4}.
	
	Finally, using $b_{ic} = r_i - p_{ic}$ for every $c \in A_i \cap W$, the inequalities from the second row read
	\[
	(u_i(W) + 1) r_i -  \sum_{c \in A_i \cap W}b_{ic} = (u_i(W) + 1) r_i - \sum_{c \in A_i \cap W} (r_i - p_{ic}) = r_i + \sum_{c \in A_i \cap W} p_{ic} \ge \alpha+\delta=1 
	\]
	for every $i \in N$. To achieve \ref{cond:p1}, we normalize ($r_i,p_{i},b_{i}$) by $r_i + \sum_{c \in A_i \cap W} p_{ic}\ge 1$.
	As a result, the inequalities in the second row of the dual turn into equalities and \ref{cond:p1} is satisfied. Again, all other inequalities (and previously deduced properties) remain feasible by homogeneity and the fact that the chosen normalization never increases payments.
	Hence, there exists a payment system $(r,p) \in \mathcal P_{1234}(W)$. 
	
	\medskip
	
	For the reverse direction, suppose we are given a payment system $(r,p) \in \mathcal P_{1234}(W)$. We show that $W$ satisfies core+. Assume for contradiction that there exists a fractional blocking coalition $y \in \mathbb R_{\ge 0}^n$ with a beneficial deviation $z \in \mathbb R_{\ge 0}^m$ (see \Cref{def:core+-deviation}). We can assume w.l.o.g. that $(y,z) \in [0,1]^n\times [0,1]^m$ since all inequalities in \Cref{def:core+-deviation} are homogeneous under non-negative scaling. The proof structure closely follows the arguments given in the main body that $\mathcal P_{1234}(W)\neq \emptyset$ is a sufficient condition for $W$ being in the core.
	We will again compute the total amount of payments to candidates under the deviation in two different ways.
	Denote by $S$ the set of voters with $y_i>0$ and by $T$ the set of candidates with $z_c>0$.
	
	First, take the perspective of the voters in $S$. Assume that voter $i$ pays $p_{ic} z_c$ if $c \in A_i \cap W$ and $r_i z_c$ if $c \in A_i \setminus W$. Since $z_c=0$ for $c \notin T$, voter $i$'s payments to candidates outside of $T$ are zero. By \ref{cond:p3} and \ref{cond:p4}, payments to a candidate $c \in T$ sum up to at most $q z_c$.
	In total, voter $i$ pays at least 
	\begin{align*}
		\sum_{c \in A_i \setminus W} r_i  z_c+\sum_{c \in A_i \cap W} p_{ic} z_c
		&{}\overset{\mathclap{(\ref{eq:fraccoredeviation})}}{\ge} r_i y_i+r_i y_i\lvert A_i \cap W\rvert -r_i \sum_{c \in A_i \cap W}\min\{y_i,z_c\}+\sum_{c \in A_i \cap W} p_{ic} z_c
		\\
		& = r_i y_i+r_i \sum_{c \in A_i \cap W}(y_i-z_c)_+ +\sum_{c \in A_i \cap W} p_{ic} z_c
		\\
		&{}\overset{\mathclap{\ref{cond:p2}}}{\ge} r_i y_i+\sum_{c \in A_i \cap W} p_{ic}\left((y_i-z_c)_+ +z_c \right)
		\\
		& \ge r_i y_i+\sum_{c \in A_i \cap W} p_{ic}y_i
		\\
		&{}\overset{\mathclap{\ref{cond:p1}}}{=} y_i,
	\end{align*} 
	where the first equality follows from $y_i-\min\{y_i,z_c\}= (y_i-z_c)_+$.
	Hence, all voters in $S$ together pay at least $\sum_{i \in N}y_i$, yielding a lower bound for the total payments to candidates in $T$.
	
	Second, take the perspective of the candidates in $T$. A candidate $c \in T \cap W$ collects a total payment of at most $\sum_{i \in N_c} p_{ic}z_c$ from voters in $S$, which is at most $q z_c$ by \ref{cond:p4}. A candidate $c \in T \setminus W$ collects a total payment of at most $\sum_{i \in N_c} r_i z_c$ from voters in $S$, which is strictly less than $q z_c$ by \ref{cond:p3}. Since there is at least one such candidate due to \eqref{eq:fraccoredeviation}, the total payments to candidates in $T$ are strictly less than $q \sum_{c \in C}z_c$.
	
	Combining the two estimates yields $\sum_{i \in N}y_i \le \text{total payments to $T$} < q \sum_{c \in C}z_c$, which contradicts that $y$ can afford $z$.
\end{proof}

\subsection{Relationship to Lindahl Priceability, Frugal Lindahl Priceability, and FJR+}

To prove the claimed implications from \Cref{fig:core-hierarchy}, we first relate core+ to \emph{Lindahl priceability}, a market-based stability notion introduced by \citet{munagala2022auditing}, and to a strengthening called \emph{frugal Lindahl priceability} introduced by \citet{golz2026apportionment}. We will state both axioms at an arbitrary quota $q$.
In the respective works, these notions have been defined using the Hare quota $q = \frac{n}{k}$.

A \emph{Lindahl price system} assigns each voter $i \in N$ a personalized price $\pi_{ic}\ge 0$ for every candidate $c \in C$. Each voter is assigned one unit of budget which she can split between the different candidates.

\begin{definition}[\citealp{munagala2022auditing}]
	A committee $W$ is \emph{Lindahl priceable} at quota $q$ if there exist non-negative personalized prices $(\pi_{ic})_{i\in N,c\in C}$ such that
	\begin{enumerate}
		\renewcommand{\theenumi}{(L\arabic{enumi})}
		\renewcommand{\labelenumi}{\theenumi}
		\item \label{eq:L1} $\sum_{i\in N}\pi_{ic}\le q$ for every $c\in C$, and
		\item \label{eq:L2} for every voter $i\in N$ and every committee $T\subseteq C$,
		\[        
		u_i(T)>u_i(W) \quad\implies\quad \sum_{c\in T}\pi_{ic}>1.
		\]
	\end{enumerate}
\end{definition}

Thus, there exists no voter that can afford a set of candidates that she strictly prefers to $W$, while the collective price for every candidate does not exceed the quota $q$. Note that, without loss of generality, we can set $\pi_{ic}=0$ whenever $c\notin A_i$. Lindahl priceability implies the core \citep{munagala2022auditing}, and indeed it is strictly stronger than the core \citep{BTC+26a}.

Interestingly, Lindahl priceability admits an equivalent definition, which is very similar to core+. This has already been observed by \citet{munagala2022auditing} but we still provide a proof within our framework for the sake of completeness. For a committee $W \subseteq C$, a \emph{fractional Lindahl blocking coalition} consists of a voter participation vector $y \in \mathbb{R}^n_{\ge 0} \setminus \{\mathbf 0\}$ and candidate quantities $z \in \mathbb{R}^m_{\ge 0}$ such that the coalition can afford the candidates: $\sum_{i \in N} y_i \ge q\sum_{c \in C} z_c$, and every member strictly prefers $z$ to $W$:
\begin{equation}\label{eq:fraccoredeviationLindahl}
	\sum_{c \in A_i} \min\{y_i,z_c\} \ge \left(u_i(W)+1\right)y_i \qquad \text{for every } i \in N.
\end{equation}

\begin{proposition}
	\label{prop:lindahl-eq-def}
	A committee $W \in \mathcal{W}$ is Lindahl priceable if and only if it admits no fractional Lindahl blocking coalition.
\end{proposition}

\begin{proof}
	Let $W \subseteq C$ be a committee.
	For the direction from right to left, assume that $W$ admits no fractional Lindahl blocking coalition.
	Note that compared to \eqref{eq:fraccoredeviation}, the respective condition for a fractional Lindahl blocking coalition takes the minimum with $y_i$ for \emph{all} $c \in A_i$, rather than only $c \in A_i \cap W$; it is otherwise identical. Hence, we can introduce additional auxiliary variables $x_{ic}$ for $c \in A_i \setminus W$ and set up its primal and dual as in the proof of \Cref{thm:p1234-core}.
	
	\[
	\begin{array}{ccc}
		\hline
		& \text{primal} & \text{dual} \\
		\hline
		\\[-0.8em]
		\forall i \in N, c \in A_i: &  x_{ic}  \ge  0   & b_{ic}+p_{ic}-r_i  \ge  0 \\[0.5em]
		
		\forall i \in N: &  y_i  \ge  0 &  \hspace{-3em}(u_i(W) + 1) r_i - \sum_{c \in A_i} b_{ic}- \alpha -\delta  \ge  0 \\[0.5em]
		
		\forall c \in C: &  z_c  \ge  0  & -\sum_{i \in N_c}p_{ic}+q\alpha  \ge  0 \\[0.5em]
		
		&  -\sum_{i \in N}y_i  =  -1 &  \delta  >  0 \\[0.5em]
		\forall i \in N, c \in A_i:&  x_{ic}-y_i  \le  0 &  b_{ic}  \ge  0 \\[0.5em]
		\forall i \in N, c \in A_i:&  x_{ic}-z_c  \le  0 &  p_{ic}  \ge  0 \\[0.5em]
		&   q\sum_{c \in C} z_c-\sum_{i \in N} y_i  \le  0 &  \alpha  \ge  0 \\[0.5em]
		\forall i \in N:&   \left(u_i(W)+1\right)y_i-\sum_{c \in A_i} x_{ic}  \le  0 &  r_i  \ge  0 \\[0.5em]
		\hline
	\end{array}
	\]
	
	Since $W$ admits no fractional Lindahl blocking coalition, the primal is infeasible. By Farkas' lemma, the dual admits a feasible solution $(\delta,b,p,\alpha,r)$. We modify that solution into a price system satisfying \ref{eq:L1} and \ref{eq:L2} which proves that $W$ is Lindahl priceable. To this end, let $t\coloneqq \alpha+\delta/2$ ($t>0$ since $\delta>0$) and define personalized prices as
	\[
	\pi_{ic}=
	\begin{cases}
		\frac{p_{ic}}{t},             & c\in A_i\\
		0,                  & c\notin A_i.
	\end{cases}
	\]
	For every $c \in C$, this yields
	\[
	\sum_{i \in N_c}\pi_{ic}=\sum_{i \in N_c}\frac{p_{ic}}{t}\le q \frac{\alpha}{\alpha+\delta/2}\le q,
	\]
	where the first inequality follows from the third row of the dual. Hence, the chosen price system satisfies \ref{eq:L1}.
	
	Let $T \subseteq C$ be a committee and $i \in N$ a voter with $u_i(T)>u_i(W)$. Then,
	\begin{align*}
	t\sum_{c \in T}\pi_{ic}=\sum_{c \in A_i \cap T}p_{ic}
	\ge  \sum_{c \in A_i \cap T}r_i -\sum_{c \in A_i \cap T}b_{ic} 
	\ge u_i(T)r_i-\sum_{c \in A_i}b_{ic} 
	\ge \alpha+\delta
	>t,
	\end{align*}
	where the first inequality follows from the first row of the dual and the third inequality is implied by the second row of the dual. Dividing by $t$ establishes \ref{eq:L2}.
	
	For the reverse direction, let $\pi$ be a price system witnessing Lindahl priceability of $W$ and assume for contradiction that there exists a fractional Lindahl blocking coalition. Due to homogeneity, we can assume that $(y,z) \in [0,1]^n\times [0,1]^m$. For every voter $i \in N$ with $y_i>0$,
	\begin{align*}
	\lvert A_i \rvert\ge \sum_{c \in A_i}\frac{\min\{y_i,z_c\}}{y_i} \ge u_i(W)+1,
	\end{align*}
	due to \eqref{eq:fraccoredeviationLindahl}, meaning that we can define $A^\pi_i \subseteq A_i$ as the set of the $u_i(W)+1$ cheapest candidates according to $\pi_i$. By \ref{eq:L2}, $\sum_{c \in A^\pi_i}\pi_{ic}>1$.
	For the cost of $z$ to each voter $i$ with $y_i>0$, we therefore obtain
	\begin{align*}
		\sum_{c \in C}\pi_{ic}z_c \ge \sum_{c \in A_i}\pi_{ic} \min\{y_i,z_c\}
		\ge y_i \sum_{c \in A_i}\pi_{ic} \frac{\min\{y_i,z_c\}}{y_i}
		\ge y_i\sum_{c \in A^\pi_i}\pi_{ic} 
		> y_i.
	\end{align*}
	The third inequality holds because $\sum_{c \in A_i}\min\{y_i,z_c\}/y_i \ge u_i(W)+1$ and $\min\{y_i,z_c\}/y_i\le 1$: The cheapest bundle of candidates with total mass at least $u_i(W)+1$ where we cannot buy more than $1$ of a candidate costs exactly $\sum_{c \in A^\pi_i}\pi_{ic}$.
	
	Summing over all voters with $y_i>0$ gives 
	\[
	\sum_{i \in N}y_i<\sum_{i \in N : y_i>0}\sum_{c \in C}\pi_{ic}z_c \le \sum_{i \in N}\sum_{c \in C}\pi_{ic}z_c\le q \sum_{c \in C}z_c,
	\]
	where the last inequality follows from \ref{eq:L1}. However, this contradicts the fact that coalition $y$ can afford $z$. 
\end{proof}

\emph{Frugal Lindahl priceability} strengthens the above definition by additionally requiring that an approved winner is never more expensive than an approved loser.
We formalize this in the next definition.

\begin{definition}[Frugal Lindahl priceability]
	A committee $W$ is \emph{frugal Lindahl priceable} at quota $q$ if it is Lindahl priceable and additionally satisfies 
	\begin{enumerate}
		\renewcommand{\theenumi}{(L\arabic{enumi})}
		\renewcommand{\labelenumi}{\theenumi}
		\setcounter{enumi}{2}
		\item \label{eq:L3} $\pi_{iw}\le \pi_{ic}$ for every $i\in N$, $w\in A_i\cap W$, and $c\in A_i\setminus W$.
	\end{enumerate}
\end{definition}

We now show that this is equivalent to core+.

\begin{proposition}
	\label{prop:core+-frugal-lindahl}
	A committee $W \in \mathcal W$ satisfies core+ at quota $q$ if and only if it satisfies frugal Lindahl priceability at quota $q$.
\end{proposition}

\begin{proof}
	Suppose $W$ satisfies core+. By \Cref{thm:p1234-core}, there exists a payment system $(r,p)$ satisfying \ref{cond:p1}--\ref{cond:p4}.
	For every losing candidate $c\in C\setminus W$, we can choose $\varepsilon_c>0$ sufficiently small such that $\sum_{i\in N_c}(r_i+\varepsilon_c)\le q$ by \ref{cond:p3}.
	This will be needed to show the strict inequality in \ref{eq:L2}.
	With this, we define personalized prices as
	\[
	\pi_{ic}=
	\begin{cases}
		p_{ic},             & c\in A_i\cap W,\\
		r_i+\varepsilon_c,  & c\in A_i\setminus W,\\
		0,                  & c\notin A_i.
	\end{cases}
	\]
	Condition \ref{eq:L1} follows from \ref{cond:p4} for winners and \ref{cond:p3} for losers. Moreover, using \ref{cond:p2} we obtain 
	\[
	\pi_{iw} = p_{iw} \le r_i<r_i+\varepsilon_c = \pi_{ic},
	\]
	for every $w \in A_i \cap W$ and $c \in A_i \setminus W$. Hence, \ref{eq:L3} holds.
	
	Finally, we show \ref{eq:L2}. Fix some voter $i \in N$ and a candidate set $T \subseteq C$ with $u_i(T)>u_i(W)$. Define $S\coloneqq A_i\cap T \cap W$, $L \coloneqq A_i\cap(T\setminus W)$, and $M=(A_i\cap W)\setminus S$. Thus, $S$ contains voter $i$'s approved candidates that belong to both $W$ and $T$, $L$ contains her approved candidates newly introduced by $T$, and $M$ contains her approved candidates from $W$ that are omitted from $T$.
	Since $u_i(T)>u_i(W)$, we get that $\lvert L \rvert \ge \lvert M\rvert+1$.
	Now by $A_i\cap T = L \cup S$, $A_i \cap W = S \cup M$ and conditions \ref{cond:p1} and \ref{cond:p2}, 
	\begin{align*}
		\sum_{c\in T}\pi_{ic}
		& \ge \sum_{w\in S}p_{iw} + \lvert L \rvert r_i+\sum_{c\in L}\varepsilon_c\\
		& > \sum_{w\in S}p_{iw} + (\lvert M \rvert + 1)r_i\\
		& \ge \sum_{w\in S}p_{iw} + \sum_{w\in M}p_{iw} + r_i = 1.
	\end{align*}   
	Hence, \ref{eq:L2} holds, and $W$ is frugal Lindahl priceable.
	
	\medskip
	Conversely, suppose that $W$ is frugal Lindahl priceable via the price system $\pi$. If $A_i \setminus W=\emptyset$, set $r_i=1$ and $p_{iw}=0$ for every $w \in A_i \cap W$. For these voters, \ref{cond:p1} and \ref{cond:p2} hold; their reserves do not appear in \ref{cond:p3}, and their payments contribute nothing to winners in \ref{cond:p4}. 
	For all other voters, set 
	\[
	\rho_i \coloneqq \min_{c\in A_i\setminus W} \pi_{ic},
	\qquad 
	s_i \coloneqq \rho_i + \sum_{w\in A_i\cap W}\pi_{iw}.
	\]
	
	If $c_i$ is a loser approved by voter $i$ with $\pi_{ic_i}=\rho_i$, then $(A_i\cap W) \cup \{c_i\}$ is strictly preferred to $W$.
	We get
	\[
	1\overset{\ref{eq:L2}}{<}s_i=\rho_i+\sum_{w \in A_i \cap W}\pi_{iw} \overset{\ref{eq:L3}}{\le}(u_i(W)+1)\rho_i.
	\]
	In particular, $\rho_i>0$.
	With this, we can now define the price allocation $(r_i, p_i)$ for every $i \in N$ and show that it satisfies conditions \ref{cond:p1}--\ref{cond:p4}:
	\[
	p_{iw} = \frac{\pi_{iw}}{s_i} \quad \text{for } w\in A_i\cap W, \qquad 
	r_i = \frac{\rho_i}{s_i}.
	\]
	
	As we scaled by $s_i$, we directly get $r_i+\sum_{w\in A_i\cap W}p_{iw}=1$, so \ref{cond:p1} holds.
	Since $s_i>1$, $p_{iw}\le\pi_{iw}$. Thus, \ref{eq:L1} implies $\sum_{i\in N_w}p_{iw}\le q$ for every candidate $w \in W$; so \ref{cond:p4} holds.
	Again using \ref{eq:L3}, we get $\pi_{iw}\le\rho_i$, and therefore, $p_{iw}\le r_i$; hence, \ref{cond:p2} holds.
	Finally, let $c\in C \setminus W$ and assume $N_c \neq \emptyset$; otherwise \ref{cond:p3} holds directly. By definition, for every $i\in N_c$,
	\[
	r_i = \frac{\rho_i}{s_i} < \rho_i\le\pi_{ic} \quad \text{gives}\quad
	\sum_{i\in N_c}r_i < \sum_{i\in N_c}\pi_{ic} \le q,
	\]
	where the strict inequality $r_i < \rho_i$ follows from $s_i > 1$, and the last inequality uses \ref{eq:L1}. This gives \ref{cond:p3}. By \Cref{thm:p1234-core}, $W$ satisfies core+.
\end{proof}

To show that core+ implies $\mathrm{FJR+}$, we note that every ``fractional violation'' \citep[][Definition 3.1]{teh2026strengthening} of $\mathrm{FJR+}$ constitutes a fractional blocking coalition with a beneficial deviation according to \Cref{def:core+-deviation} after normalizing $\sum_{i \in N}y_i$ to $1$. A fractional violation of $\mathrm{FJR+}$ is a feasible solution to the primal given in the proof of \Cref{thm:p1234-core}. In particular, the inequalities from the last row of the primal hold due to the fact that every voter $i$ with $y_i>0$ has utility $u_i(W)\le \ell-1$ but their participating $y_i$-fraction receives at least $\ell y_i$ when deviating. In contrast to \citet{teh2026strengthening}, we defined auxiliary variables $x_{ic}$ only for $c \in A_i \cap W$ and worked directly with $z_c$ (instead of $x_{ic}$) for $c \in A_i \setminus W$.

\subsection{From $\mathrm{EJR}$ to $\mathrm{EJR+}$ via fractional deviations}

In this section, we illustrate that $\mathrm{EJR}$ and $\mathrm{EJR+}$, defined for arbitrary quotas $q \in (\frac{n}{k+1},\frac{n}{k}]$, admit a connection similar to core and core+ (as well as $\mathrm{FJR}$ and $\mathrm{FJR+}$ \citep{teh2026strengthening}). To this end, we define yet another type of blocking coalition.

\begin{definition}\label{def:ejr+-deviation}
	For a committee $W \subseteq C$ and $\ell \in [k]$, a \emph{fractional homogeneous blocking coalition} consists of a voter participation vector $y \in \mathbb{R}^n_{\ge 0} \setminus \{\mathbf 0\}$ with $u_i(W)\le \ell-1$ for all voters $i$ with $y_i>0$, and candidate quantities $z \in \mathbb{R}^m_{\ge 0}$ with $\supp(z)\subseteq \bigcap_{i : y_i>0}A_i$ such that the coalition can afford the candidates: $\sum_{i \in N} y_i \ge q\sum_{c \in C} z_c$, and every member strictly prefers $z$ to $W$:
	\begin{equation}\label{eq:fracejrdeviation}
		\sum_{c \in A_i \setminus W} z_c +\sum_{c \in A_i \cap W} \min\{y_i,z_c\} \ge \ell y_i \qquad \text{for every } i \in N.
	\end{equation}
\end{definition}

\begin{proposition}\label{prop:ejr+-connection-ejr}
	Let $W \subseteq C$ be a committee. Then,
	\begin{enumerate}
	\item existence of a fractional homogeneous blocking coalition with $(y,z)\in \{0,1\}^n \times \{0,1\}^m$ is equivalent to an $\mathrm{EJR}$ violation, and
	\item existence of a fractional homogeneous blocking coalition is equivalent to an $\mathrm{EJR+}$ violation. 
	\end{enumerate}
\end{proposition}

\begin{proof}
	For the first part, note that a fractional homogeneous blocking coalition for $\ell \in[k]$ with $(y,z)\in \{0,1\}^n \times \{0,1\}^m$ corresponds to a set of voters $S=\supp(y)$ and a set of candidates $T=\supp(z)$ such that 
	\[
	u_i(W) \le \ell-1 \text { for all } i \in S, \quad T \subseteq \bigcap_{i \in S}A_i, \quad \text{and } \lvert S \rvert\ge q \lvert T \rvert.
	\]
	Moreover, \eqref{eq:fracejrdeviation} yields
	\[
	\lvert T \rvert=\sum_{c \in A_i \setminus W} z_c +\sum_{c \in A_i \cap W} \min\{1,z_c\} \ge \ell.
	\]
	This yields an $\mathrm{EJR}$ violation. 
	
	Conversely, any EJR violation can be used to construct such a blocking coalition by taking the respective coalition and $\ell$ of their commonly approved candidates to construct $y$ and $z$, respectively.
	\medskip
	
	For the second part, recall that an $\mathrm{EJR+}$ violation \citep{BrPe23a} consists of $\ell \in [k]$, a candidate $c \in C \setminus W$, and a set of voters $S\subseteq N_c$ such that
	\[
	u_i(W)\le \ell-1 \text{ for all }i \in S \quad \text{and }\lvert S \rvert \ge q\ell.
	\] 
	Given a fractional homogeneous blocking coalition, set $S=\supp(y)$ and let $i^* \in \argmax_{i \in S} y_i$. First, $u_i(W)\le \ell-1$ for all $i \in S$ by assumption. Second, \eqref{eq:fracejrdeviation} for voter $i^*$ implies $\sum_{c \in C}z_c \ge \ell y_{i^*}$. Hence,
	\[
	q \ell y_{i^*}\le q \sum_{c \in C}z_c \le \sum_{i \in S}y_i \le \lvert S \rvert y_{i^*}
	\]
	where the second inequality follows from affordability. Thus, $\lvert S \rvert\ge q \ell$.
	Third, $\supp(z) \not\subseteq W$ as otherwise $u_i(W)y_i<\ell y_i$ for every $i \in S$ by \eqref{eq:fracejrdeviation}. Hence, there exists a candidate $c^* \not\in W$ that is approved by all voters in $S$ since $z_{c^*}>0$. This gives an $\mathrm{EJR+}$ violation. 
	
	Conversely, assume we are given an $\mathrm{EJR+}$ violation with $\ell$, $c^* \not\in W$, and $S$. Set $y_i=\frac{1}{\ell}$ for all $i \in S$ and $y_i=0$ otherwise. Furthermore, choose $z_{c^*}=1$ and $z_c=0$ otherwise. Then, $u_i(W)\le \ell-1$ and $c^* \in \bigcap_{i : y_i>0}A_i$ by assumption. Moreover, $\sum_{i \in N}y_i= \frac{\lvert S \rvert}{\ell}\ge q$ and the left-hand side of \eqref{eq:fracejrdeviation} equals $1=\ell y_i$ for $i \in S$ by construction. For $i \not\in S$, the left-hand side is non-negative and the right-hand side is $0$. Thus, $(y,z)$ is a fractional homogeneous blocking coalition.
\end{proof}

\citet[][Theorem 3.3]{teh2026strengthening} already states an equivalence of $\mathrm{EJR+}$ violations to fractional homogeneous blocking coalitions with $\lvert \supp(z) \rvert=1$ instead of $\supp(z) \subseteq \bigcap_{i : y_i>0}A_i$. \Cref{prop:ejr+-connection-ejr} establishes equivalence of these two conditions for that setting.

\section{Properties of the Robin Hood Operator} \label{appendix:robin-hood}

In this section, we establish the properties of the Robin Hood operator stated in \Cref{section:robin-hood}.
Recall that for $x \in \Delta^d$ and integer $\ell \ge 1$, the operator $\Phi^\ell$ introduces $\ell$ new coordinates at a common water level $t_\ell$ determined by 
\[
	\ell \cdot t_\ell = \sum_{j=1}^d(x_j-t_\ell)_+.
\]
The resulting distribution is obtained by capping the original coordinates at $t_\ell$ and assigning this value to each new coordinate. As we now show, it is thereby equivalent to apply the single-coordinate Robin Hood operator $\Phi$ successively $\ell$ times. 

\begin{lemma} \label{lem:robin-hood-iteration}
    For every $x\in\Delta^d$ and integer $\ell\geq 1$, we have $\Phi^{\ell+1}(x) =\Phi(\Phi^\ell(x))$.
\end{lemma}

\begin{proof}
	Fix some distribution $x \in \Delta^d$, and let $t_\ell$ and $t_{\ell+1}$ be the corresponding water levels of $\Phi^\ell(x)$ and $\Phi^{\ell+1}(x)$.
	
	We first show that $t_{\ell+1} < t_{\ell}$. To this end, we can observe that the function $t \mapsto \sum_{j=1}^d (x_j - t)_+ - \ell \cdot t$ is strictly decreasing and equals zero at $t_\ell$. Hence, replacing $\ell$ with $\ell+1$ makes its value at $t_\ell$ negative, and therefore, we must have that $t_{\ell+1}$ is strictly smaller. 

	We now apply $\Phi$ to $\Phi^\ell(x)$. Since the latter distribution consists of $d$ coordinates $\min\{x_j, t_\ell\}$ and $\ell$ additional coordinates equal to $t_\ell$, the water level when applying $\Phi$ is determined by 
	\[
		s = \sum_{j=1}^d \bigl(\min\{x_j, t_\ell\} - s \bigr)_+ + \ell(t_\ell - s)_+.
	\]
	We claim that $s = t_{\ell+1}$.
	We make use of the two water-level equations for $\Phi^\ell$ and $\Phi^{\ell+1}$ and the fact that $t_{\ell+1}<t_\ell$. At $t_{\ell+1}$, the right-hand side of the above equation equals
	\begin{align*}
		\sum_{j=1}^d \bigl(\min\{x_j, t_\ell\} - t_{\ell+1} \bigr)_+ + \ell (t_\ell - t_{\ell + 1})_+ 
		= & \sum_{j=1}^d \bigl(x_j - t_{\ell+1} \bigr)_+ - \sum_{j= 1}^d \bigl(x_j - t_{\ell} \bigr)_+ + \ell(t_\ell - t_{ \ell + 1})_+ \\
		= & (\ell+1)t_{\ell + 1} - \ell t_{\ell} + \ell(t_\ell -t_{\ell+ 1}) \\
		= & t_{\ell +1 },
	\end{align*}
	where the first equality uses the fact that $(\min\{a,b\} - c)_+ = (a-c)_+ - (a-b)_+$, whenever $c \leq b$. By uniqueness of the water level, we can conclude that $s=t_{\ell+1}$.
	As a result, applying $\Phi$ caps each existing coordinate at $t_{\ell+ 1}$ and sets all $\ell+1$ new coordinates to this same value. This is precisely the definition of $\Phi^{\ell+1}(x)$, which proves the statement.
\end{proof}

We next formalize the claim that the Robin Hood operator produces the most balanced extension of a given distribution. To compare probability distributions according to how concentrated their mass is, we use the standard notion of \emph{majorization}. Given a distribution $x \in \Delta^d$, let $\bar x$ denote the vector obtained by sorting its coordinates in non-decreasing order.

\begin{definition}
	A distribution $x \in \Delta^d$ is majorized by $x'\in \Delta^d$, written $x \prec x'$, if 
	\[
		\sum_{j=1}^s \bar x_j \geq \sum_{j=1}^s \bar x'_j  \qquad\text{for all }s\in[d].
	\]
\end{definition}

Intuitively, $x \prec x'$ means that the mass of $x$ is at least as evenly distributed as that of $x'$. We show that the Robin Hood operator produces a distribution majorized by every distribution that does not increase the mass of any existing coordinate.
For $x \in \Delta^d$ and $\ell \ge 1$, let $\mathcal{Q}_\ell(x) \coloneqq \{ q \in \Delta^{d+\ell} : q_j \leq x_j \text{ for all } j \in [d]\}$ denote the set of feasible extensions obtained by adding $\ell$ additional coordinates.

\begin{proposition}
	For every $x\in\Delta^d$ and integer $\ell\geq1$, $\Phi^\ell(x) \prec q$ for all $q \in \mathcal Q_\ell(x)$.
\end{proposition}

\begin{proof}
	Fix distributions $x \in \Delta^d$ and $q \in \mathcal{Q}_\ell(x)$ and let $\ell\ge1$. Write $y = \Phi^\ell(x)$ and $t=t_\ell$ for the resulting distribution and water level, respectively.
	We partition the coordinates into $A \coloneqq \{j \in [d] : x_j < t\}$ and $B \coloneqq [d+\ell]\setminus A$ with $a = \lvert A \rvert$ and $b=\lvert B \rvert$.
	Thus, $A$ contains the smallest coordinates of $y$, while the remaining coordinates in $B$ are uniform due to the (repeated) application of the Robin Hood operator.

	For $s \leq a$, the $s$ smallest coordinates of $y$ are in $A$. Choose their indices $I_s \subseteq A$. By feasibility of $q$, we obtain
	\[
		\sum_{j=1}^s \bar q_j \le \sum_{j \in I_s}q_j \le \sum_{j \in I_s}x_j = \sum_{j=1}^s \bar y_j.
	\]
	For $a < s < d+\ell$, we instead need to compare the $d+\ell - s$ largest coordinates. Feasibility and total mass give 
	\[
		\sum_{j \in B}q_j = 1 - \sum_{j \in A}q_j \geq 1 - \sum_{j \in A}y_j = bt.
	\]
	Hence, the average of the coordinates of $q$ in $B$ is at least $t$. The $d+\ell - s$ largest coordinates of $q$ in $B$ have average at least this large, and their sum is at most that of the $d+\ell-s$ largest coordinates of $q$ overall. Thus,
	\[
        \sum_{j=s+1}^{d+\ell} \bar q_j \geq (d+\ell-s)t = \sum_{j=s+1}^{d+\ell} \bar y_j.
    \]
	Subtracting from the common total mass $1$ yields the desired inequality for the $s$ smallest coordinates.
	For $s = d+\ell$, both sums equal $1$, and therefore, $y = \Phi^\ell(x) \prec q$.
\end{proof}

Lastly, we prove the characterization of the water level $t_\ell(x)$ given a distribution $x \in \Delta^d$ and step size $\ell \ge 0$. This will be useful both for establishing the properties of harmonic entropy as well as formulating a linear programming approximation.

\waterlevelCharacterization*

\begin{proof}
	Fix some distribution $x \in \Delta^d$. For $\ell=0$, both representations equal $\max_{j \in [d]}x_j$. To see this, observe that the largest average over a non-empty subset is the largest coordinate and that $\sum_j(x_j-t)_+\leq0$ holds precisely when $t \ge \max_{j \in [d]} x_j$. Hence, we can focus on the case $\ell \ge 1$. 

	To simplify notation, set $t = t_\ell(x)$. By the water level equation, $\sum_{j=1}^d (x_j - t)_+ = \ell \cdot t$. Therefore, for every non-empty subset $S \subseteq [d]$, we get
	\[
		\sum_{j \in S} x_j \leq \lvert S \rvert t + \sum_{j \in S}(x_j - t)_+ \leq \lvert S \rvert t + \ell \cdot t.
	\]
	Hence, rearranging gives a lower bound on the water level:
	\[
		t\geq\frac{\sum_{j\in S}x_j}{\lvert S \rvert + \ell}.
	\]
	Note that this bound is attained when $S$ consists of the coordinates above the water level. To see this, let $S^* \coloneqq \{j \in [d] : x_j > t\}$. The water-level equation again gives
	\[
		\sum_{j\in S^*}x_j = \lvert S^* \rvert t + \sum_{j\in S^*}(x_j - t) =(\lvert S^* \rvert + \ell)t.
	\]
	Hence, the water level $t_\ell(x)$ can be written as the maximum over linear functions, i.e., 
	\[
    	t_\ell(x) = \max_{\emptyset\neq S\subseteq[d]} \frac{\sum_{j\in S}x_j}{\lvert S \rvert + \ell}.
	\]

	For the second representation, we show that its constraint holds exactly at and above the water level. Consider the auxiliary function $h(s) = \sum_{j=1}^d(x_j-s)_+ - \ell s$ with $s \ge 0$. Note that this function is strictly decreasing and $h(t)=0$. Hence
	\[
        \sum_{j=1}^d(x_j-s)_+\leq\ell s \quad\Longleftrightarrow\quad s\geq t.
    \]
	Intuitively, below the water level, the mass above $s$ exceeds the total capacity $\ell s$ of the new coordinates. Thus, the water level $t$ is the smallest feasible value of $s$, proving the second representation.
\end{proof}

In particular, $t_\ell$ is continuous and convex, and $\frac{1}{d+\ell}\leq t_\ell(x)\leq\frac{1}{\ell+1}$ for every distribution $x \in \Delta^d$.

\section{Properties of Harmonic Entropy} \label{appendix:harmonic-entropy}

We start by deriving the infinite-series representation of harmonic entropy stated in \eqref{equation:harmonic-entropy-infinite-series}. We then bound the error incurred by truncating this series. Together with the water level characterization shown in \Cref{appendix:robin-hood}, these results establish the axiomatic properties stated in \Cref{prop:harmonic-entropy-properties}.

\begin{lemma} \label{lem:harmonic-entropy-equivalence}
	Let $\Phi^0(x)=x$. For every $x \in \Delta^d$,
	\[
		F(x) = \sum_{\ell=0}^{\infty} \left(\frac{1}{\ell+1} - \lVert \Phi^\ell(x) \rVert_\infty\right).
	\]
\end{lemma}

\begin{proof}
	Fix $x \in \Delta^d$. We denote by $r$ the number of positive coordinates.
	By our previous arguments, we know that after sufficiently many applications of the Robin Hood operator $\Phi$, the positive coordinates become and remain uniform. Thus, for all sufficiently large $L$, the positive coordinates of $\Phi^L(x)$ are uniform on $r+L$ coordinates, giving $F(\Phi^L(x)) = H_{r+L-1}$. Iterating the Robin Hood operator and using \Cref{lem:robin-hood-iteration} gives
	\begin{align*}
		F(x) &= F(\Phi^L(x)) - \sum_{\ell=0}^{L-1}\lVert \Phi^\ell(x) \rVert_\infty \\
			 &= \sum_{\ell=0}^{L-1} \left(\frac{1}{\ell+1} -\lVert \Phi^\ell(x) \rVert_\infty\right) + H_{r+L-1} - H_L.
	\end{align*}
	Note that the second equality extends the sum by the harmonic number $H_L$.
	The difference $H_{r+L-1} - H_L = \sum_{h=1}^{r-1}\frac{1}{L+h} \leq \frac{r-1}{L}$ tends to zero as $L$ goes to infinity. Taking this limit proves the claimed representation as an infinite series.
\end{proof}

Recall that the truncated version of the infinite-series formulation of harmonic entropy is given by 
\[
	F_L(x) = \sum_{\ell=0}^{L-1} \left(\frac{1}{\ell+1} - \lVert \Phi^\ell(x) \rVert_\infty \right).
\]
The next lemma bounds the tail independently of the distribution $x$. This will allow us to deduce continuity of $F$ from the continuity of its partial sums. We will later also use this bound to control the approximation error when computing committee values.

\begin{lemma}\label{lem:harmonic-entropy-tail}
    For every $x\in\Delta^d$ and integer $L\geq1$, we have $0 \leq F(x) - F_L(x) \leq \frac{d-1}{L}$.
\end{lemma}

\begin{proof}
	By definition, $F(x) - F_L(x) = \sum_{\ell=L}^\infty (\frac{1}{\ell+1}- \lVert\Phi^\ell(x)\rVert_\infty)$.
	Since $\Phi^\ell(x)$ is a probability vector of dimension $d + \ell$ and at least $\ell+1$ coordinates equal the maximum, we get
	\[
		\frac{1}{d+\ell} \leq \lVert\Phi^\ell(x)\rVert_\infty \leq \frac{1}{\ell+1}.
	\]
	Consequently, each summand is non-negative and bounded above by $\frac{1}{\ell+1} -\frac{1}{d+\ell}$.
	Summing over all entries in $F(x) - F_L(x)$ from $L$ to $\infty$ yields a telescoping sum upper bounding the tail by
	\[
		\sum_{\ell=L}^\infty \left(\frac{1}{\ell+1} - \frac{1}{d+\ell}\right) 
		= \sum_{\ell = L}^{L+d-2}\frac{1}{\ell+1} 
		= \sum_{r=1}^{d-1}\frac{1}{L+r}
		\leq \frac{d-1}{L}.
	\]
\end{proof}

We can finally establish the claimed properties of $F$. The water-level characterization makes each partial sum continuous and concave, while the tail bound ensures that the convergence is uniform.

\harmonicentropyproperties*

\begin{proof}
	By \Cref{lem:water-level-characterization}, each water level
	\[
		t_\ell(x) = \max_{\emptyset\neq S\subseteq[d]} \frac{\sum_{j\in S}x_j}{\lvert S \rvert + \ell}
	\]
	is a maximum over finitely many linear functions of $x$. It is therefore continuous and convex. Consequently, every summand must be continuous and concave, as is every partial sum $F_L$.

	To prove \emph{concavity} of $F$, let $x, q \in \Delta^d$ and $\lambda \in [0,1]$. By concavity of $F_L$, we get $F_L(\lambda x + (1-\lambda)q) \geq \lambda F_L(x) + (1-\lambda) F_L(q)$. Letting $L \rightarrow\infty$ proves the same inequality for $F$.

	Similarly, for \emph{continuity}, \Cref{lem:harmonic-entropy-tail} gives
    \[
        \sup_{x\in\Delta^d}  F(x)-F_L(x) \leq \frac{d-1}{L} \rightarrow 0.
    \]
	Hence, the continuous functions $F_L$ converge uniformly to $F$ in $\Delta^d$, and therefore, $F$ is continuous.

	Furthermore, it is easy to see that permuting coordinates of $x$ only permutes the subsets in the maximum formulation of the water level, and hence, leaves it unchanged. This gives \emph{symmetry} of $F$.
	Finally, invariance under adding zero coordinates holds by definition.
\end{proof}

We finally relate harmonic entropy to Shannon entropy, which admits an analogous representation in terms of \emph{continuous} water levels.

\begin{lemma} \label{lem:shannon-water-level}
    For every $x \in \Delta^d$, the continuous water levels satisfy
    \[
        \int_0^\infty \left( \frac{1}{s+1} - t_s(x) \right) \,\mathrm{d}s = -\sum_{j \in [d]} x_j \ln x_j,
    \]
    with the convention $0\ln0 = 0$.
\end{lemma}

\begin{figure}[ht!]
	\centering
    \begin{tikzpicture}[
        x=2.2cm, y=10.5cm,
        font=\small,
        line cap=round, line join=round,
        >=Latex,
        axis/.style={->,draw=black!65,line width=0.5pt},
        guide/.style={draw=black!32,line width=0.45pt,densely dashed}
    ]
        \definecolor{waterkeep}{RGB}{45,105,145}
        \definecolor{waternew}{RGB}{230,145,45}

        \pgfmathsetmacro{\baselevel}{1/7}
        \pgfmathsetmacro{\sliceheight}{1/4}
        \pgfmathsetmacro{\slicelength}{4/3}

        \fill[black!9] (0,0) rectangle (4,\baselevel);
        \fill[waterkeep!14]
            (0,\baselevel) -- (0,1/2)
            -- plot[domain=0:0.5,samples=30,variable=\s]
                (\s,{1/(2*(\s+1))})
            -- plot[domain=0.5:3,samples=65,variable=\s]
                (\s,{5/(6*(\s+2))})
            -- plot[domain=3:4,samples=25,variable=\s]
                (\s,{1/(\s+3)})
            -- cycle;

        \draw[guide] (0,\baselevel) -- (4,\baselevel);
        \draw[guide] (4,0) -- (4,\baselevel);

        \draw[waterkeep,line width=1.05pt]
            plot[domain=0:0.5,samples=30,variable=\s]
                (\s,{1/(2*(\s+1))})
            plot[domain=0.5:3,samples=65,variable=\s]
                (\s,{5/(6*(\s+2))})
            plot[domain=3:4,samples=25,variable=\s]
                (\s,{1/(\s+3)});

        \draw[waternew!75!black,line width=0.55pt,densely dashed]
            (\slicelength,0) -- (\slicelength,\sliceheight);
        \draw[waternew,line width=1.55pt]
            (0,\sliceheight) -- (\slicelength,\sliceheight);
        \fill[waternew] (\slicelength,\sliceheight) circle[radius=1.8pt];

        \draw[axis] (0,0) -- (4.25,0) node[right] {$s$};
        \draw[axis] (0,0) -- (0,0.55) node[above] {$t$};
        \node[below left=2pt] at (0,0) {$0$};
        \node[left=5pt] at (0,1/2) {$\lVert x\rVert_\infty$};
        \node[left=5pt] at (0,\sliceheight) {$t$};
        \node[left=5pt] at (0,\baselevel) {$a$};
        \node[below=5pt] at (\slicelength,0) {$g(t)$};
        \node[below=5pt] at (4,0) {$M$};
        \node at (2.45,{\baselevel/2}) {$Ma$};
        \node[text=waterkeep!80!black] at (0.34,0.325) {$D$};
        \node[text=waterkeep!80!black,above=5pt]
            at (2.75,{5/(6*(2.75+2))}) {$t_s(x)$};
        \node[below=4pt] at ({\slicelength/2},\sliceheight)
            {\strut length $g(t)$};
    \end{tikzpicture}
	\caption{The continuous water-level curve for $x=(\frac{1}{2},\frac{1}{3},\frac{1}{6})$ and $M=4$, with $a=t_M(x)=\frac{1}{7}$.
	The gray rectangle has area $Ma$. The blue region $D$ can be measured using vertical slices of length $t_s(x) - a$ or horizontal slices of length $g(t)$. The orange segment shows one such horizontal slice.}
	\label{fig:continuous-water-level}
\end{figure}

\begin{proof}
	Assume without loss of generality that $x \in \Delta^d$ is strictly positive. Otherwise, we could just delete the zero coordinates from both sides without affecting the water level or Shannon entropy.
	Let $M>0$ be sufficiently large such that $a = \frac{1}{d+ M}< \min_j x_j$. The continuous water-level then gives $\sum_j (x_j-a)_+ = 1 - da = Ma$. By uniqueness, this establishes $t_M(x) = a$. 

	For $0 < t \le \lVert x \rVert_\infty$, define the inverse function of $s \mapsto t_s(x)$ by $g(t) \coloneqq \sum_{j \in [d]} (\frac{x_j}{t} -1)_+$ which gives $g(t_s(x))=s$ (see also \Cref{fig:continuous-water-level} for a visualization). Using this inverse relation, we obtain the following when integrating over the water level $t_s(x)$:
	\begin{align*}
		\int_0^M t_s(x)\,\mathrm{d}s
		& = Ma + \int_0^M \int_a^{t_s(x)} 1 \,\mathrm{d}t\,\mathrm{d}s = Ma + \int_a^{\lVert x \rVert_\infty}g(t)\,\mathrm{d}t \\
        & = Ma + \sum_{j \in [d]} \int_a^{\lVert x \rVert_\infty} \left(\frac{x_j}{t}-1\right)_+\,\mathrm{d}t = Ma + \sum_{j \in [d]} \int_a^{x_j} \left( \frac{x_j}{t} - 1 \right)\,\mathrm{d}t \\
        & = \sum_{j \in [d]} x_j\ln x_j - \ln a + (M+ d)a - 1 = \ln(M +d) + \sum_{j \in [d]} x_j \ln x_j,
	\end{align*}
	where the last equality uses $x \in \Delta^d$ and $(M+d)\; a = 1$. Consequently, using $\int_0^M \frac{1}{s+ 1}\,\mathrm{d}s = \ln(M + 1)$, we get
	\[
        \int_0^M \left(\frac{1}{s + 1} - t_s(x) \right)\,\mathrm{d}s = \ln \frac{M + 1}{M + d} - \sum_{j \in [d]} x_j \ln x_j.
    \]
	Taking $M$ to infinity proves the stated identity.
\end{proof}

\section{Properties of the Harmonic Entropy Rule}
\label{appendix:further-harmonic-entropy-results}

\subsection{Axiomatics of the Harmonic Entropy Rule}

\heroParetoOptimal*
\begin{wrapstuff}[
    r,
    top=4,
    width=0.36\textwidth,
    hsep=1.5em,
    vsep=0.6em
]
    \begin{minipage}{\linewidth}
    \centering
    \definecolor{waterkeep}{RGB}{45,105,145}

    \begin{tikzpicture}[
        x=1cm,
        y=0.75cm,
        candidate/.style={
            circle,
            draw,
            minimum size=6mm,
            inner sep=0pt,
            font=\small
        },
        approved/.style={
            draw=waterkeep,
            fill=waterkeep!10,
            text=waterkeep,
            thick
        },
        excluded/.style={
            draw=gray!35,
            fill=gray!10,
            text=gray!50
        },
        matching/.style={
            -{Latex[length=1.5mm]},
            draw=waterkeep,
            thick
        }
    ]
        \node at (0,1) {$A_i$};
        \node at (2.9,1) {$A_i$};

        \foreach \c in {1,3}
            \node[candidate, approved](l\c) at (0,1-\c) {$c_{\c}$};

        \foreach \c in {2,4}
            \node[candidate, excluded](l\c) at (0,1-\c) {$c_{\c}$};

        \foreach \c in {2,3,4}
        	\node[candidate, approved](r\c) at (2.9,1-\c) {$c_{\c}$};

        \node[candidate, excluded] (r1) at (2.9,0) {$c_1$};

        \draw[matching] (l1) -- node[above, sloped, font=\small] {$p_{ic_1}$}(r2);
        \draw[matching] (l3) -- node[above, font=\small] {$p_{ic_3}$}(r3);
        \node[anchor=east, font=\small, xshift=-1mm]at (r4.west) {$p'_{ic_4}=0$};
    \end{tikzpicture}

    \captionsetup{type=figure}
    \captionof{figure}{Matching $A_i\cap W$ into $A_i\cap W'$.
        Gray nodes are outside $W$ (left) or $W'$ (right).}
    \label{fig:pareto-payment-matching}
    \end{minipage}
\end{wrapstuff}
\begin{proof}
	Let $W$ be a committee selected by the harmonic entropy rule and $(r, p)$ an optimal payment system for $W$. Note that, by \Cref{thm:local-opt-satisfy-core+}, $(r, p)$ satisfies \ref{cond:p3}.
	Now suppose that a committee $W' \in \mathcal W_k$ Pareto-dominates $W$. We construct a payment system $(r',p') \in \mathcal{P}_{124}(W')$ which has a strictly higher harmonic entropy value, contradicting our assumption.
	
	For each voter $i$, set her reserve $r'_i = r_i$ and payments $p'_{ic} = p_{ic}$ to candidates $c \in A_i \cap W \cap W'$. Injectively match the candidates in $A_i \cap (W \setminus W')$ to candidates in $A_i \cap (W' \setminus W)$ and transfer each matched candidate's payment to its image; set all remaining payments to zero. Such a matching exists since $u_i(W') \ge u_i(W)$ for every $i \in N$; see \Cref{fig:pareto-payment-matching}. Since each voter $i$'s payment is unchanged up to permutation and zero coordinates, $\sum_{i \in N}F(r_i, p_i) = \sum_{i \in N}F(r'_i, p'_i)$  and $(r', p')$ satisfies \ref{cond:p1} and \ref{cond:p2}.
	Furthermore, every new candidate $c \in W' \setminus W$ receives at most the reserve $r_i$ from each voter $i$. By \ref{cond:p3}, $\sum_{i\in N_c} p'_{ic} \leq \sum_{i \in N_c} r_i < q$. Thus, the new payment system additionally satisfies \ref{cond:p4}.

	Choose a voter $j$ with $u_j(W')>u_j(W)$. After the reassignment, this voter has at least one candidate $c^* \in A_j \cap (W'\setminus W)$ with $p'_{jc^*}=0$. By \ref{cond:p3}, $\sum_{i\in N_{c^*}} p'_{ic^*} \leq \sum_{i \in N_{c^*}} r_i-r_j < q-r_j$. We can apply $\Phi$ to this voter's distribution $(r'_j, p'_j)$, treating this new coordinate as the payment to $c^*$ and deleting the zero coordinate. By the properties of $\Phi$, all payment conditions remain feasible. In particular, $c^*$ does not receive more than $r_j$ from voter $j$, so \ref{cond:p4} still holds for $c^*$. 
	By the Robin Hood equality, this step increases voter $j$'s entropy by $r'_j>0$, while all other voters' entropies remain unchanged. Consequently,
	\[
		\mathcal{V}(W') \ge \mathcal{V}(W)+r_j >\mathcal V(W),
	\]
	contradicting $W$ being a maximizer of the harmonic entropy rule. 
\end{proof}

\begin{example}[Harmonic entropy rule violates priceability]
	Consider an instance with $n=13$ voters and committee size $k = 2$, with 5 voters approving $\{a,b\}$, 5 voters approving $\{a,c\}$, and 3 voters approving $\{d\}$. Use the Hare quota $q = \frac{n}{k} = \frac{13}{2}$. One can check that on this instance, the harmonic entropy rule agrees with PAV and uniquely selects $W = \{a,d\}$. However, $W$ is not priceable in the sense of \citet{PeSk20a}. Suppose it was, with price $p > 0$. Then $p \le 3$ because only 3 voters can pay for $d \in W$. The first 10 voters can only pay for $a$, so they have a total remaining budget of $10-p$. But $b$ and $c$ are losers, so their supporters each must hold at most $p$ remaining budget, hence $10- p \le 2p$ which implies $p \ge 10/3 > 3$, a contradiction. \qed 
\end{example}

\subsection{Comparison to \pav}

It is possible to define \pav as optimizing the harmonic entropy objective but over payment systems in $\mathcal P_{123}(W)$ instead of $\mathcal P_{124}(W)$. To make this formal, define the following objective functions similar to the objective function $\mathcal{V}$ used by the harmonic entropy rule:
\[
\mathcal{V}_{123}(W) = \max_{(r,p) \in \mathcal P_{123}(W)} \sum_{i \in N} F(r_i,p_i),
\]
where the maximum is interpreted as $-\infty$ when $\mathcal P_{123}(W) = \emptyset$. Define $\mathcal{V}_{12}(W)$ and $ \mathcal{V}_{1234}(W)$ analogously.

\begin{lemma}
	\pav can be defined equivalently as selecting the committees $W \in \mathcal W_k$ maximizing $\mathcal{V}_{123}(W)$.
\end{lemma}
\begin{proof}
	Let $W$ be a committee with induced utilities $u_i = \lvert A_i \cap W \rvert$, and let $(r,p)$ be its uniform payment system. By symmetry and concavity, $F$ is maximized on $\Delta^{u_i + 1}$ by the uniform distribution. Since $(r,p) \in \mathcal P_{12}(W)$ and $\mathcal P_{123}(W) \subseteq \mathcal P_{12}(W)$, we have
	\[
	\mathcal{V}_{123}(W) \le \mathcal{V}_{12}(W) = \sum_{i\in N} F(\tfrac{1}{u_i + 1}, \dots, \tfrac{1}{u_i + 1}) = \sum_{i \in N} H_{u_i} = \pav(W),
	\]
	with equality whenever $(r,p) \in \mathcal P_{123}(W)$. By \Cref{thm:pav-p123}, every \pav-optimal committee $W^*$ satisfies $(r,p) \in \mathcal P_{123}(W^*)$, and thus, achieves $\mathcal{V}_{123}(W^*) = \pav(W^*)$, whereas any non-\pav-optimal committee $W$ has $\mathcal{V}_{123}(W) \le \pav(W) < \pav(W^*) = \mathcal{V}_{123}(W^*)$. Thus, $\mathcal{V}_{123}$ and \pav have the same maximizers.
\end{proof}

Another way to relate the harmonic entropy rule and \pav is via their respective score functions. Indeed, the harmonic entropy function $F$ is maximized at uniform distributions, where its value is a harmonic number (see the discussion after \Cref{prop:harmonic-entropy-properties}). We can use this fact to relate the objective function $\mathcal{V}$ of the harmonic entropy rule and the \pav score as follows.
\begin{lemma}
	\label{lem:pav-score-higher-than-harmonic-entropy-score}
	For every committee $W \in \mathcal{W}$, we have $\mathcal V(W)\le \pav(W)$, with equality whenever the uniform payment system of $W$ satisfies \ref{cond:p4}.
\end{lemma}
\begin{proof}
	Let $W \in \mathcal{W}$ and fix a voter $i \in N$. By symmetry and concavity of $F$, harmonic entropy is maximized over $\Delta^{u_i(W)+1}$ at the uniform distribution with value $H_{u_i(W)}$. Hence,
	\[
	\mathcal V(W) = \max_{(r, p) \in \mathcal P_{124}(W)}  \sum_{i \in N} F(r_i, p_i) \le \sum_{i \in N} H_{u_i(W)} = \pav(W).
	\]
	Whenever the uniform payment system $(\bar r, \bar p)$ of $W$ satisfies \ref{cond:p4}, then $(\bar r, \bar p) \in \mathcal P_{124}(W)$, and $\sum_{i \in N} F(\bar r_i, \bar p_i) = \sum_{i \in N} H_{u_i(W)} = \pav(W)$, so we have $\mathcal V(W) = \pav(W)$.
\end{proof}

In particular, \Cref{lem:pav-score-higher-than-harmonic-entropy-score} lets us show that the harmonic entropy rule coincides with \pav on party-approval instances. Two candidates $c, c' \in C$ are \emph{copies} if $N_c = N_{c'}$, i.e., if they are approved by the same voters. An instance is a \emph{party-approval instance} if every candidate has at least $k$ copies.

\heroPartyApproval*
\begin{proof}
	Let $W$ be a committee selected by \pav, and let $(r,p)$ be its uniform payment system. By \Cref{thm:pav-p123}, $(r,p) \in \mathcal{P}_{123}(W)$. We show that $(r,p)$ also satisfies \ref{cond:p4}, i.e., $\sum_{i \in N_c} p_{ic} \le q$ for all $c \in W$. Fix some $c \in W$. We distinguish two cases.
	\begin{itemize}
		\item Suppose there exists a copy $c'$ of $c$ that is a loser, i.e., $c' \not \in W$. Then because $(r,p)$ satisfies \ref{cond:p3}, we have $\sum_{i \in N_c} p_{ic} = \sum_{i \in N_{c'}} r_i < q$, where we used that $(r,p)$ is uniform.
		\item Otherwise, $W$ consists only of copies of $c$, and hence every approver of $c$ has utility $k$, giving
		\[
		\sum_{i\in N_c}p_{ic}
		=\frac{\lvert N_c \rvert}{k+1}\le\frac n{k+1}<q.
		\]
	\end{itemize}
	In either case, we have established \ref{cond:p4} for $c$. By the equality case of \Cref{lem:pav-score-higher-than-harmonic-entropy-score}, we have $\mathcal V(W) = \pav(W)$. Now, for \emph{every} committee $W' \in \mathcal{W}$, we have
	\[
	\mathcal V(W')\le \pav(W') \le\pav(W)=\mathcal V(W),
	\]
	where the second inequality is strict whenever $W'$ is not \pav-optimal. Hence the \pav-optimal committees are exactly the committees maximizing $\mathcal{V}$.
\end{proof}

Note that the proof of \Cref{prop:hero-party-approval} shows that on party-approval instances, the two rules not only select the same committees, but they also induce the same objective values, i.e., $\pav(W)=\mathcal V(W)$ for optimal committees $W$. This makes it easy to prove an NP-hardness result for the harmonic entropy rule.

\heroNPComplete*
\begin{proof}
	\citet[Theorem 5.1]{BGP+24a} show that it is NP-complete to decide, given a party-approval instance and a number $B \ge 0$, whether there exists a committee $W$ with $\pav(W) \ge B$. But by our discussion above, this is the exact same decision problem as the one described in the statement of the proposition.
\end{proof}

Note that all the results we discussed in this section apply for any quota $\frac{n}{k+1} < q \le \frac{n}{k}$.

\section{Computing $\mathcal{V}(W)$ Exactly} \label{appendix:computing-entropy-objective}

We have seen that it suffices to set $L = \lceil 2nk/\delta \rceil$ in order to obtain a core+ committee with our local-search algorithm. We now show that the linear program in \eqref{eq:harmonic-entropy-LP}, together with a correction term, computes $\mathcal{V}(W)$ for $W \in\mathcal W$ \emph{exactly} when we set $L = (k+1)^4$.

We start again by analyzing the truncated infinite-series formulation of harmonic entropy.
To this end, for a distribution $x \in \Delta^d$ and integer $L \ge 1$, let the tail of the infinite series $F$ be defined as
\[
	\mathcal{T}_L(x) \coloneqq F(x) - F_L(x) = \sum_{\ell = L}^\infty \bigl( \frac{1}{\ell+1} - \lVert \Phi^\ell(x) \rVert_\infty \bigr).
\] 
If every coordinate of $x$ is at least $\frac{1}{L}$, then the water level after $L$ steps lies below every initial coordinate. Hence, the resulting distribution is uniform and remains uniform thereafter (see \Cref{section:robin-hood}). The tail therefore depends only on $d$ and $L$.

\begin{proposition} \label{prop:tail-formula}
	Let $x \in \Delta^d$ and let $L \ge 1$ be an integer. If $x_j \geq \frac{1}{L}$ for every $j \in [d]$, then $\mathcal{T}_L(x) = \sum_{h=1}^{d-1} \frac{1}{L+h}$.
\end{proposition}

\begin{proof}
	The water-level bound yields
	\[
		t_L(x) \leq \frac{1}{L+1}< \frac{1}{L} \leq x_j \qquad \text{for all } j \in [d],
	\]
	where the first inequality holds because the $L$ new coordinates and at least one original coordinate all equal the water level $t_L(x)$, and their total mass is at most $1$.
	From this bound, it follows that $\Phi^L(x)$ is uniform on all $d+L$ coordinates. Since the Robin Hood operator preserves uniformity, we get $t_\ell(x) = \frac{1}{d+\ell}$ for all $\ell \ge L$. Consequently,
	\[
		F(x) - F_L(x) = \sum_{\ell = L}^{\infty} \left(\frac{1}{\ell +1} - \frac{1}{d+  \ell}\right) = \sum_{h = 1}^{d-1} \frac{1}{L + h},
	\]
	where the last equality follows from telescoping -- $\frac{1}{d+\ell}$ appears with a negative sign at index $\ell$ and with a positive sign at index $\ell+d-1$.
\end{proof}

It remains to find an optimizer of our truncated harmonic entropy objective whose coordinates are all bounded away from zero to apply \Cref{prop:tail-formula}.

For the proof, we use an exchange argument between two voters: one increases her relatively small payment for a candidate, while another decreases her large payment for the same one. To show that such an exchange improves the objective, we first bound how transfers between coordinates affect the water level. For $j \in [d]$, let $e_j$ be the $j$th standard basis vector.

\begin{lemma} \label{lem:entropy-increase-via-exchanges}
	Let $x \in \Delta^d$, $a, b \in [d]$ be distinct coordinates, and $0<\varepsilon\leq x_a$.
	Set $x' = x + \varepsilon (e_b - e_a)$ and fix some integer $\ell \ge 0$.
	\begin{enumerate}
		\item[(a)] If $x_a - \varepsilon \ge x_b + \varepsilon$, then $t_\ell(x') \leq t_\ell(x)$. If additionally $x_a - \varepsilon \geq t_\ell(x) \geq x_b + \varepsilon$, then $t_\ell(x') \leq t_\ell(x) - \frac{\varepsilon}{d + \ell}$.
		
		\item[(b)] If $x_b = \lVert x \rVert_\infty$, then $t_\ell(x') \leq t_\ell(x) +  \frac{\varepsilon}{\ell+1}$. If additionally $x_a - \varepsilon \geq t_\ell(x)$, then $t_\ell(x') = t_\ell(x)$. 
	\end{enumerate}
\end{lemma}

\begin{proof}
	Fix $t = t_\ell(x)$. From \Cref{lem:water-level-characterization}, we know that $s \geq 0$ is an upper bound on the water level $t_\ell(x')$ whenever $\sum_{j=1}^d (x'_j - s)_+ \le \ell s$. 
	
	\medskip
	(a): For the first transfer, we move mass from a larger coordinate to a smaller one without reversing their order. Hence, by convexity of $u \mapsto (u-t)_+$, 
	\[
		(x_a - \varepsilon - t)_+ + (x_b + \varepsilon - t)_+ \le (x_a - t)_+ + (x_b - t)_+.
	\]
	Since all other coordinates are unchanged, $\sum_{j \in [d]}(x'_j - t)_+ \leq \ell t$, and therefore, $t_\ell(x') \leq t$.
	Under the additional assumption that the donor $a$ remains above or at the water level $t$, while the recipient $b$ remains below or at it, the transfer reduces the total excess by $\varepsilon$, i.e., $\sum_j(x'_j - t)_+ = \ell t - \varepsilon$. 
	Set $s = t - \frac{\varepsilon}{d+\ell} \ge 0$. Non-negativity follows from $t \ge x_b + \varepsilon \ge \varepsilon$. Lowering the water level from $t$ to $s$ increases the excess of each coordinate by at most $t-s$, that is, $(x_j' -s )_+ = (x_j' - t + t - s)_+ \le (x_j' - t)_+ + (t-s)$.
	Summing over all $d$ coordinates gives
	\[
		\sum_{j \in [d]}(x_j' - s)_+ \le \ell t - \varepsilon + d(t-s) = \ell s,
	\]
	where the last equality follows from our choice $t - s = \frac{\varepsilon}{d + \ell}$.
	Thus, $t_\ell(x')\leq s = t - \frac{\varepsilon}{d+\ell}$.

	\medskip
	(b): We now assume that the recipient $b$ is a largest coordinate, i.e., $x_b = \lVert x \rVert_\infty$. Let $s = t + \frac{\varepsilon}{\ell + 1}$.
	Since $x_b \ge t$ and $s-t \le \varepsilon$, we have $(x_b' - s)_+ = (x_b - t)_+ + \varepsilon - (s-t) \ge 0$. For every other coordinate, the mass weakly decreases and the threshold increases, so $(x_j' - s)_+ \le (x_j - t)_+$. Summing over all coordinates gives
	\[
		\sum_{j \in [d]}(x'_j - s)_+ \leq \ell t + \varepsilon - (s-t) = \ell s,
	\]
	where the equality uses $s-t = \frac{\varepsilon}{\ell + 1}$. Hence, $t_\ell(x') \leq t + \frac{\varepsilon}{\ell + 1}$.
	Under the additional assumption that both affected coordinates remain above or at the water level $t$, the total excess remains unchanged, and by uniqueness of the water level, $t_{\ell}(x') = t$ for $\ell\ge 1$. For $\ell=0$, the additional assumption cannot hold, since $x_a \le t_0(x)$ and $\varepsilon>0$. This establishes the statement.
\end{proof}

The lemma allows us to compare the gain from increasing a small payment with the loss incurred by decreasing a large one. 
We use this comparison to show that a suitable maximizer has every coordinate bounded below by an explicit positive constant which only depends on $k$.

\begin{proposition} \label{prop:interior-maximizer}
	Let $W\in \mathcal W_k$ be a committee and fix an integer $L \ge (k+1)^3$. Among all maximizers of $\sum_i F_L(r_i,p_i)$ over $\mathcal P_{124}(W)$, choose $(r,p)$ minimizing $\sum_i r_i$. For every $i \in N$, we then have $r_i \geq \frac{1}{k+1}$ and $p_{ic}\geq \frac{1}{(k+1)^4}$ for every $c \in A_i \cap W$.
\end{proposition}

\begin{proof}
	Let $(r, p)$ maximize $\sum_{i \in N} F_L(r_i,p_i)$ over $\mathcal P_{124}(W)$, breaking ties by minimizing the total reserve $\sum_{i \in N}r_i$. Note that such a payment system always exists (see the proof of \Cref{lem:de-reserved-system}).
	By \ref{cond:p1} and \ref{cond:p2}, we have $1 \le (\lvert A_i \cap W \rvert+1)r_i \leq (k+1)r_i$. Hence, $r_i \ge \frac{1}{k+1}$ for every voter $i \in N$.
	Since $F_L$ is symmetric and concave, the same exchange argument as in the proof of \Cref{lem:de-reserved-system} establishes that $p_{ic} < r_i$ for some voter $i$ and candidate $c \in A_i \cap W$ implies $\sum_{i \in N_c}p_{ic}=q$.

	Now suppose for a contradiction that $p_{ic}< \frac{1}{(k+1)^4}$ for some $i \in N$ and $c \in A_i \cap W$. Together with $r_i \ge \frac{1}{k+1}$, it follows that $p_{ic} < r_{i}$ and the candidate must be fully funded.
	Hence, there must exist a different voter $i' \in N_c$ with $p_{i'c} \ge \frac{q}{\lvert N_c \rvert} \ge \frac{q}{n} > \frac{1}{k+1}$.
	We now apply our exchange argument for these two voters $i$ and $i'$. Voter $i'$ moves some sufficiently small positive mass $\varepsilon>0$ from $c$ to her reserve, while $i$ moves the same amount from her reserve to $c$.
	As a result, the truncated harmonic entropy objective increases, contradicting maximality. It remains to show this strict increase in the objective.

	For voter $i$ and every step size $1 \le \ell \le (k+1)^3-1$, the maximum representation of \Cref{lem:water-level-characterization} yields
	\[
		t_\ell(r_i, p_i) \ge \frac{r_i}{\ell+1} \ge \frac{1}{(k+1)^4} > p_{ic}.
	\]
	Moreover, for $\ell\ge 1$, $t_\ell(r_i, p_i) < r_i$ as $r_i$ is a largest coordinate.
	Together with the preceding lower bound, this allows us to choose $\varepsilon>0$ so that 
	\[
		r_i - \varepsilon \ge t_\ell(r_i, p_i) \ge p_{ic}+\varepsilon \qquad \text{for all } 1\le \ell \le (k+1)^3 - 1. 
	\] 
	\Cref{lem:entropy-increase-via-exchanges}(a) then guarantees that each of these water levels drops by at least $\frac{\varepsilon}{k+1+\ell}$. All other water levels weakly decrease, and therefore, the objective increases by at least
	\[
		\varepsilon \sum_{\ell = 1}^{(k+1)^3-1}\frac1{(k+1+\ell)} =\varepsilon \bigl( H_{(k+1)^3+k} - H_{k+1} \bigr).
	\]

	For voter $i'$, mass moves to the reserve, which is a largest coordinate. We choose $\varepsilon$ small enough that $p_{i'c} - \varepsilon \ge \frac{1}{k+1}$. For every $\ell \ge k$, we have $t_\ell(r_{i'}, p_{i'}) \le \frac{1}{\ell+1} \le \frac{1}{k+1} \leq p_{i'c} - \varepsilon$. Thus, \Cref{lem:entropy-increase-via-exchanges}(b) shows that these water levels are unchanged, while the remaining water levels increase by at most $\frac{\varepsilon}{\ell+1}$. This voter's objective therefore decreases by at most 
	\[
		\varepsilon\sum_{\ell=0}^{k-1}\frac1{\ell + 1} = \varepsilon H_k.
	\]

	We now compare the derived gain and loss. Since $z \mapsto \frac{1}{z}$ is decreasing, for $k \ge 2$, we get
	\begin{align*}
		H_{(k+1)^3+k}-H_{k+1} &\geq \int_{k+2}^{(k+1)^3+k+1} \frac{1}{z} \,\mathrm{d}z\\
        					  &= \ln \frac{(k+1)((k+1)^2+1)}{k+2}\\
        					  & \geq \ln \bigl( k (k+1) \bigr) > 1 + \ln k \geq H_k.
	\end{align*}
	The second inequality uses that $((k+1)^2+1)/(k+2)\geq k$, the strict inequality uses $\ln(k+1) > 1$ for $k\ge 2$, and the last one uses $H_k \leq 1 + \ln k$. For $k=1$, the strict inequality is also satisfied as $H_9 - H_2 > H_1$.
	Thus, the exchange strictly increases the total objective.
	This change preserves \ref{cond:p1} and \ref{cond:p4}, but may violate \ref{cond:p2} for voter $i$ as the reserve might not be a largest coordinate anymore. Swapping the reserve with a largest coordinate restores \ref{cond:p2}. Note that $F_L$ is unchanged due to symmetry.
	As a result, we obtain a contradiction to the maximality of the initial payment system and obtain $p_{ic} \ge \frac{1}{(k+1)^4}$.
\end{proof}

We can finally combine this coordinate bound with the tail formula from \Cref{prop:tail-formula}.
More specifically, a suitable optimizer of the truncated objective attains the tail bound with equality, providing an exact formula for $\mathcal V(W)$.

\begin{theorem} \label{thm:exact-computation-entropy}
	Let $W \in \mathcal W_k$ and set $L=(k+1)^4$. Let $(r,p)$ maximize $\sum_{i \in N} F_L(r_i, p_i)$ over $\mathcal P_{124}(W)$ with minimum total reserve among the maximizers. Then $\mathcal V(W) = \mathcal{V}_L(W) + \sum_{i \in N}\sum_{h=1}^{\lvert A_i \cap W \rvert}\frac{1}{L+h}$ and we can compute $\mathcal{V}(W)$ exactly in polynomial time for rational $q$.
\end{theorem}

\begin{proof}
	Let $(r, p) \in \mathcal P_{124}(W)$ maximize $\sum_{i \in N} F_L(r_i, p_i)$ with minimum total reserve among the maximizers. Since $L=(k+1)^4 \geq (k+1)^3$, \Cref{prop:interior-maximizer} ensures that $r_i \ge \frac{1}{k+1}$ and $p_{ic} \ge \frac{1}{(k+1)^4}$ for every $i \in N$ and $c \in A_i \cap W$. By \Cref{prop:tail-formula}, we obtain $\mathcal T_L(r_i,p_i)
    = \sum_{h=1}^{\lvert A_i\cap W \rvert}\frac{1}{L+h}$, with $d = \lvert A_i\cap W \rvert + 1$. 
	For every voter $i \in N$, the tail formula therefore attains its universal upper bound
	\[
		\mathcal T_L(\hat r_i, \hat p_i) \leq \sum_{h=1}^{\lvert A_i\cap W \rvert}\frac{1}{L + h} = \mathcal T_L(r_i,p_i),
	\]
	for every feasible payment system $(\hat r,\hat p) \in \mathcal P_{124}(W)$. Note that the universal upper bound follows by summing $t_\ell(x) \ge \frac{1}{d+\ell}$ over $\ell \ge L$, as in the proof of \Cref{lem:harmonic-entropy-tail}.
	Together with the optimality of the truncated objective, this gives
	\[
		\sum_{i \in N} F(\hat r_i,\hat p_i) \leq \mathcal V_L(W) + \sum_{i \in N}\mathcal T_L(r_i, p_i) = \sum_{i \in N} F(r_i, p_i).
	\]
	Hence, $(r, p)$ also maximizes the full objective, and $\mathcal V(W) = \mathcal V_L(W)+ \sum_{i\in N}\sum_{h=1}^{\lvert A_i\cap W \rvert} \frac{1}{L+h}$.
	
	Lastly, for rational $q$, the linear program in \eqref{eq:harmonic-entropy-LP} has $O(nk^5)$ variables and constraints. Its optimum can therefore be computed in polynomial time. The correction term can likewise be computed in polynomial time, establishing the statement.
\end{proof}

\section{Discussion of Existing Voting Rules} \label{appendix:discussion}

In this section, we establish the entries in \Cref{tab:payment-certificates}. Throughout, we fix the quota $q = \frac{n}{k}$. Note that MES and GJCR denote their original versions, which may return fewer than $k$ candidates.
We start by providing two characterizations that will simplify the analysis.

\begin{lemma} \label{lem:payment-characterizations}
	Let $W \subseteq C$ be a committee.
	\begin{enumerate}
		\item[(a)] $\mathcal{P}_{13}(W) \neq \emptyset$ if and only if $W$ satisfies JR.
		\item[(b)] $\mathcal{P}_{123}(W) \neq \emptyset$ if and only if $\sum_{i \in N_c}\frac{1}{u_i(W)+1}<q$ for every $c \notin W$.
	\end{enumerate}
\end{lemma}

\begin{proof}
	(a): Observe that an unrepresented voter $i$ has reserve $r_i=1$ by \ref{cond:p1}. Thus, \ref{cond:p3} requires that every unselected candidate have fewer than $q$ unrepresented supporters, which is precisely JR.
	Conversely, if $W$ satisfies JR, let each represented voter spend her entire budget on one approved committee member. Unrepresented voters retain their budget as reserve. It is easy to see that this system satisfies \ref{cond:p1} and \ref{cond:p3}.

	(b): Conditions \ref{cond:p1} and \ref{cond:p2} imply that $1 = r_i+\sum_{d\in A_i\cap W} p_{id} \leq (u_i(W)+1) r_i$. Hence, every such system has $r_i \ge \frac{1}{u_i(W)+1}$ and \ref{cond:p3} provides the stated inequality. 
	Conversely, the uniform payment system satisfies \ref{cond:p1} and \ref{cond:p2}, and by assumption, also \ref{cond:p3}.
\end{proof}

Since the existence of a payment system satisfying \ref{cond:p1} and \ref{cond:p3} is equivalent to justified representation, all listed rules except sequential \pav always admit such a payment certificate \citep{LaSk23a,BrPe23a}. This proves the entire first column of \Cref{tab:payment-certificates}.

For the second column, we have already shown that \pav, being equivalent to the uniform payment rule, satisfies conditions \ref{cond:p1}--\ref{cond:p3}. All others violate these requirements. As we have seen in \Cref{thm:p123-properties}, the three conditions imply EJR+, which is violated by sequential Phragmén. Although MES and GJCR satisfy EJR+, this property does not suffice for a payment certificate with \ref{cond:p1}--\ref{cond:p3}. We here provide a short counterexample:

\begin{figure}[ht!]
    \centering
    \definecolor{waterkeep}{RGB}{45,105,145}
    \definecolor{waternew}{RGB}{230,145,45}

    \begin{tikzpicture}[
        x=0.75cm,
        y=0.55cm,
        winner/.style={
            draw=waterkeep!70!black,
            fill=waterkeep!40,
            thick
        },
        excluded/.style={
            draw=waternew!75!black,
            fill=waternew!15,
            dashed,
            thick
        }
    ]
        \foreach \i in {1,...,6}
            \node[anchor=north] at (\i-0.5,-0.15)
                {$v_{\i}$};

        \draw[winner] (0,0) rectangle (1,1);
        \node at (0.5,0.5) {$a$};

        \foreach \y/\j in {0/1,1/2}{
            \draw[winner]
                (1,\y) rectangle (3,\y+1);
            \node at (2,\y+0.5) {$b_{\j}$};
        }

        \foreach \y/\j in {0/1,1/2,2/3}{
            \draw[winner]
                (3,\y) rectangle (6,\y+1);
            \node at (4.5,\y+0.5) {$c_{\j}$};
        }

        \draw[excluded] (0,3) rectangle (2,4);
        \node at (1,3.5) {$d$};

        \draw[excluded] (3,3) rectangle (4,4);
        \node at (3.5,3.5) {$d$};
    \end{tikzpicture}

    \caption{An instance in which MES and GJCR may return a committee admitting no payment system satisfying \ref{cond:p1}--\ref{cond:p3}.}
    \label{fig:mes-gjcr-p123-counterexample}
\end{figure}

Consider the instance in \Cref{fig:mes-gjcr-p123-counterexample} with $n=k=6$ and Hare quota $q=1$.
Both MES and GJCR might return the committee $W = \{a, b_1, b_2, c_1, c_2, c_3\}$ depicted in blue. 
MES may choose $c_1,c_2,c_3$ at price $\frac{1}{3}$ each, then $b_1,b_2$ at price $\frac{1}{2}$ each, and finally $a$ at price $1$. GJCR may select $c_1,c_2,c_3$ at level $3$, $b_1,b_2$ at level $2$, and $a$ at level $1$.
The utility vector of committee $W$ is denoted by $u(W) = (1,2,2,3,3,3)$. For the supporters of $d$, it therefore holds that 
\[
	\sum_{i \in N_d} \frac{1}{u_i(W)+1} = \frac{1}{2} + \frac{1}{3} + \frac{1}{4} = \frac{13}{12} > 1 = q.
\]
By \Cref{lem:payment-characterizations}(b), it follows that $\mathcal{P}_{123}(W) = \emptyset$. Thus, neither MES nor GJCR guarantees such a payment certificate.

Next, we analyze the voting rules with respect to the payment conditions \ref{cond:p1}, \ref{cond:p3}, and \ref{cond:p4}. For MES, its actual payments provide the required certificate. We normalize voters' budgets to $1$ and the candidate costs to the quota $q$, thereby interpreting their unspent budget as reserve. Hence, it is easy to see that \ref{cond:p1} and \ref{cond:p4} are satisfied. Since MES terminates only when no remaining candidate is affordable anymore, every unselected candidate's supporters have less than $q$ in reserve. This gives \ref{cond:p3}.

For sequential Phragmén, we can analogously use its payment interpretation: voters accumulate budget at unit rate, and a candidate is bought as soon as its supporters have accumulated the cost $q$. Let $T$ be the time when the $k$-th candidate is selected.\footnote{$T$ always exists unless less than $k$ candidates are approved by at least one voter. In that case, sequential Phragmén chooses an arbitrary superset of size $k$ of these candidates. The payment system $r_i=1$ for all $i$ then satisfies \ref{cond:p1}, \ref{cond:p3}, and \ref{cond:p4} since no unselected candidate is approved by any voter.} Since total spending is $kq=n$, we get $nT \ge n$, and hence, $T \ge 1$.
Each voter has hence accumulated a total of $T$ and so after normalization, we get \ref{cond:p1}.
Every selected candidate receives $\frac{q}{T} \leq q$, giving \ref{cond:p4}. Moreover, the supporters of every unselected candidate collectively hold at most $q$ budget, since otherwise that candidate would have been selected earlier in the process. If the collective reserve for every candidate is strictly less than $q$, \ref{cond:p3} is immediately satisfied. Hence, assume that there is an unselected candidate whose supporters hold $q$ in reserve. Then some budget remains unspent, so $nT > kq = n$, and therefore, $T > 1$. After normalization, the supporters of every unselected candidate have collective reserve at most $\frac{q}{T}<q$, establishing \ref{cond:p3}.

Lastly, GJCR selects candidates witnessing EJR+ violations. Whenever a candidate is selected, its cost $q$ is divided equally among its supporters witnessing this violation without voters exceeding their unit budget \citep{BrPe23a}. Again, interpreting their unspent budget as reserve gives a payment system satisfying \ref{cond:p1} and \ref{cond:p4}.
GJCR may, however, terminate with fewer than $k$ candidates selected and its stopping condition does not ensure that all remaining candidates are unaffordable, violating \ref{cond:p3}. A simple extension nevertheless recovers this payment condition. Starting from these payments, repeatedly add an unselected candidate whose supporters can still accumulate the cost $q$ from their reserves. At termination, no remaining candidate is affordable, which gives \ref{cond:p3}.

Finally, note that only the harmonic entropy rule and \pav satisfy \ref{cond:p1}--\ref{cond:p3} among these rules. \Cref{fig:pav-p4-core-counterexample} shows that committees chosen by \pav do not always satisfy \ref{cond:p4}, which completes the fourth column.

\end{document}